\newif\ifeightpage
\eightpagefalse

\documentclass{article}

\usepackage{amsmath}
\usepackage{amssymb}
\usepackage{amsthm}
\usepackage{bbm}

\usepackage{url}
\usepackage{xcolor}
\usepackage{hyperref}
\usepackage{doi}
\usepackage{graphicx}
\IfFileExists{draftmode}
  {\setkeys{Gin}{draft} }{}
\usepackage{subcaption}
\usepackage{microtype}
\usepackage{booktabs} \usepackage{xspace}
\usepackage{multicol}
\usepackage{multirow}

\usepackage[accepted]{icml2020}

\graphicspath{{fig/}}

\theoremstyle{plain}
\newtheorem{theorem}{Theorem}[section]
\newtheorem{lemma}[theorem]{Lemma}
\newtheorem{claim}[theorem]{Claim}
\newtheorem{corollary}[theorem]{Corollary}

\newtheorem{fact}[theorem]{Fact}

\theoremstyle{definition}
\newtheorem{definition}{Definition}[section]

\theoremstyle{remark}
\newtheorem{remark}{Remark}[section]

\newcommand{\eps}{\epsilon}

\newcommand{\R}{\mathbb{R}}

\newcommand\card[1]{\vert {#1}\vert}
\newcommand{\poly}{\operatorname{poly}}

\newcommand{\SDim}{\ensuremath{\mathrm{sdim}}\xspace}

\newcommand{\FL}{Feldman-Langberg\xspace}

\DeclareMathOperator{\Decomp}{Balanced-Decomp}
\DeclareMathOperator{\BoundReduct}{Boundary-Reduction}

\DeclareMathOperator{\tw}{tw}

\newcommand{\calB}{\mathcal{B}}

\newcommand{\calP}{{\mathcal{P}}}

\newcommand{\calS}{{\mathcal{S}}}

\newcommand{\calT}{\mathcal{T}}
\newcommand{\calV}{\mathcal{V}}

\newcommand{\minn}[1]{\min\{{#1}\}}

\providecommand{\set}[1]{{\{#1\}}}

\newcommand{\ProblemName}[1]{\textsc{#1}}

\newcommand{\kMedian}{\ProblemName{$k$-Median}\xspace}
\newcommand{\kMeans}{\ProblemName{$k$-Means}\xspace}
\newcommand{\kCenter}{\ProblemName{$k$-Center}\xspace}
\newcommand{\cost}{\ensuremath{\mathrm{cost}}\xspace}
\newcommand{\OkM}{\ProblemName{Ordered $k$-Median}\xspace}

\begin{document}
	
\icmltitlerunning{Coresets for Clustering in Graphs of Bounded Treewidth}
\twocolumn[
\icmltitle{Coresets for Clustering in Graphs of Bounded Treewidth}

\begin{icmlauthorlist}
\icmlauthor{Daniel Baker}{jhu}
\icmlauthor{Vladimir Braverman}{jhu}
\icmlauthor{Lingxiao Huang}{yale}
\icmlauthor{Shaofeng H.-C. Jiang}{wis}
\icmlauthor{Robert Krauthgamer}{wis}
\icmlauthor{Xuan Wu}{jhu}
\end{icmlauthorlist}

\icmlaffiliation{jhu}{Johns Hopkins University, USA}
\icmlaffiliation{yale}{Yale University, USA}
\icmlaffiliation{wis}{Weizmann Institute of Science, Israel}

\icmlcorrespondingauthor{Daniel Baker}{dnb@cs.jhu.edu}
\icmlcorrespondingauthor{Vladimir Braverman}{vova@cs.jhu.edu}
\icmlcorrespondingauthor{Lingxiao Huang}{huanglingxiao1990@126.com}
\icmlcorrespondingauthor{Shaofeng H.-C. Jiang}{shaofeng.jiang@weizmann.ac.il}
\icmlcorrespondingauthor{Robert Krauthgamer}{robert.krauthgamer@weizmann.ac.il}
\icmlcorrespondingauthor{Xuan Wu}{wu3412790@gmail.com}

\icmlkeywords{Coreset, Clustering, Graph, Treewidth, Road Network}

\vskip 0.3in
]
\ifeightpage
\else
	\printAffiliationsAndNotice{Authors are listed in alphabetical order.}  \fi
\begin{abstract}

We initiate the study of coresets for clustering in graph metrics,
i.e., the shortest-path metric of edge-weighted graphs. 
Such clustering problems are essential to data analysis
and used for example in road networks and data visualization. 
A coreset is a compact summary of the data that approximately preserves
the clustering objective for every possible center set, 
and it offers significant efficiency improvements 
in terms of running time, storage, and communication,
including in streaming and distributed settings. 
Our main result is a near-linear time construction
of a coreset for \kMedian in a general graph $G$,
with size $O_{\epsilon, k}(\tw(G))$ 
where $\tw(G)$ is the treewidth of $G$,
and we complement the construction with a nearly-tight size lower bound.
The construction is based on the framework of Feldman and Langberg [STOC 2011],
and our main technical contribution, as required by this framework,
is a uniform bound of $O(\tw(G))$ on the shattering dimension
under any point weights. 
We validate our coreset on real-world road networks,
and our scalable algorithm constructs tiny coresets with high accuracy,
which translates to a massive speedup of 
existing approximation algorithms such as local search for graph \kMedian.

\end{abstract}

	\section{Introduction}
\label{sec:intro}
We initiate the study of coresets for clustering in \emph{graph metrics},
i.e., the shortest-path metrics of graphs.
As usual in these contexts, the focus is on edge-weighted graphs $G=(V,E)$
with a restricted topology, and in our case bounded treewidth.
Previously, coresets were studied extensively
but mostly under geometric restrictions, e.g., for Euclidean metrics.

\paragraph{Coresets for $k$-Clustering}
We consider the \emph{metric} \kMedian problem,
whose input is a metric space $M=(V, d)$ and an $n$-point data set $X \subseteq V$,
and the goal is to find a set $C \subseteq V$ of $k$ points,
called \emph{center set}, that minimizes the objective function
\begin{equation*}
\cost(X, C) := \sum_{x \in X}{d(x, C)},
\end{equation*}
where $d(x, C):=\minn{d(x,c): c\in C}$.
The metric \kMedian generalizes the well-known Euclidean case, in which $V = \mathbb{R}^{\mathsf{d}}$ and $d(x, y) = \| x - y \|_2$.
\kMedian problem and related $k$-clustering problems
(like \kMeans, whose objective is $\sum_{x \in X}{(d(x, C))^2}$),
are essential tools in data analysis and are used in many application domains, such as genetics, information retrieval, and pattern recognition.
However, finding an optimal clustering is a nontrivial task,
and even in settings where polynomial-time algorithms are known,
it is often challenging in practice because data sets are huge,
and potentially distributed or arriving over time.
To this end, a powerful data-reduction technique, called \emph{coresets},
is of key importance.

Roughly speaking, a coreset is a compact summary of the data points
by weighted points, that approximates the clustering objective
for every possible choice of the center set.
Formally, an $\eps$-coreset for \kMedian is
a subset $D \subseteq V$ with weight $w : D \to \R_{+}$, such that
for every $k$-subset $C \subseteq V$,
\begin{equation*}
\sum_{x \in D}{w(x) \cdot d(x, C)}
\in (1 \pm \eps) \cdot \cost(X, C).
\end{equation*}
This notion, sometimes called a strong coreset,
was proposed in~\cite{HM04}, following a weaker notion of~\cite{AHV04}.
Small-size coresets (where size is defined as $|D|$) often translate
to faster algorithms, more efficient storage/communication of data,
and streaming/distributed algorithms via the merge-and-reduce framework,
see e.g.~\cite{HM04,DBLP:conf/esa/FichtenbergerGSSS13,balcan2013distributed,huang2018epsilon, DBLP:journals/siamcomp/FriggstadRS19} and recent surveys~\cite{Phillips16,DBLP:journals/ki/MunteanuS18,DBLP:journals/widm/Feldman20}.

Coresets for \kMedian were studied extensively in Euclidean spaces,
i.e., when $V = \mathbb{R}^\mathsf{d}$ and $d(x, y) = \| x - y \|_2$.
The size of the first $\epsilon$-coreset for \kMedian,
when they were first proposed \cite{HM04},
was $O(k (\frac{1}{\epsilon})^\mathsf{d} \cdot \log{n})$,
and it was improved to $O(k (\frac{1}{\epsilon})^\mathsf{d})$,
which is independent of $n$, in~\cite{DBLP:journals/dcg/Har-PeledK07}.
Feldman and Langberg~\cite{feldman2011unified}
drastically improved the dependence on the dimension $\mathsf{d}$,
from exponential to linear,
achieving an $\epsilon$-coreset of size $O(\frac{k}{\epsilon^{2}} \cdot \mathsf{d})$,
and this bound was recently generalized to doubling metrics~\cite{huang2018epsilon}.
Recently, coresets of size \emph{independent} of $\mathsf{d}$ and polynomial in $\frac{k}{\epsilon}$
were devised by~\cite{DBLP:conf/focs/SohlerW18}.

\paragraph{Clustering in Graph Metrics}
While clustering in Euclidean spaces is very common and well studied,  clustering in graph metrics is also of great importance and has many applications.
For instance, clustering is widely used for community detection in social networks~\cite{fortunato2010community}, and is an important technique for the visualization of graph data~\cite{DBLP:journals/tvcg/HermanMM00}.
Moreover, $k$-clustering on graph metrics is one of the central tasks in
data mining of spatial (e.g., road) networks~\cite{DBLP:journals/tkde/ShekharL97,DBLP:conf/sigmod/YiuM04},
and it has been applied in various data analysis methods~\cite{rattigan2007graph,cui2008geometry}, and many other applications can be found in a survey~\cite{tansel1983state}.

Despite the importance of graph \kMedian, coresets for this problem were not studied before,
and to the best of our knowledge,
the only known constructions applicable to graph metrics
are coresets for general $n$-point metrics $M=(V, d)$ with $X = V$ \cite{chen2009coresets,feldman2011unified},
that have size $\poly\log n$.
In contrast, as mentioned above,
coresets for Euclidean spaces usually have size independent of $n=\card{V}$
and sometimes even independent of the dimension $d$.
Moreover, this generic construction
assumes efficient access to the distance function,
which is expensive in graphs and requires to compute all-pairs shortest paths.

To fill this gap, we study coresets for \kMedian on the shortest-path metric of an edge-weighted graph $G$.
As a baseline, we confirm that the $O(\log n)$ factor in coreset size is really necessary for general graphs,
which motivates us to explore whether structured graphs admit smaller coresets.
We achieve this by designing coresets whose size are independent of $n$
when $G$ has a bounded \emph{treewidth} (see Definition~\ref{def:treewidth}),
which is a special yet common graph family.
Moreover, our algorithm for constructing the coresets runs
in \emph{near-linear time} (for every graph regardless of treewidth).

Indeed, treewidth is a well-studied parameter that measures how close a graph is to a tree~\cite{robertson1986graph,kloks1994treewidth},
and intuitively it guarantees a (small) vertex separator in every subgraph.
Several important graph families have bounded treewidth: trees have treewidth at most $1$,
series-parallel graphs have treewidth at most $2$,
and $k$-outerplanar graphs, which are an important special case of planar graphs, have treewidth $O(k)$.
In practice, treewidth is a good complexity measure for many types of graph data.
A recent experimental study showed that real data sets in various domains
including road networks of the US power grid networks
and social networks such as an ego-network of Facebook,
have small to moderate treewidth~\cite{DBLP:conf/icdt/ManiuSJ19}.

\subsection{Our Results}
Our main result is a near-linear time construction of a coreset for \kMedian whose size depends linearly on the treewidth of $G$ and is completely independent of $|X|$ (the size of the data set).
This significantly improves the generic $O(\frac{k}{\epsilon^2} \cdot \log n)$ size bound from~\cite{feldman2011unified}
whenever the graph has small treewidth.
\begin{theorem}[Fast Coresets for Graph \kMedian; see Theorem~\ref{thm:coreset}]
	\label{thm:coreset_intro}
For every edge-weighted graph $G=(V, E)$, $0 < \epsilon < 1$, and integer $k \geq 1$,
	\kMedian of every data set $X \subseteq V$
        (with respect to the shortest-path metric of $G$)
        admits an $\epsilon$-coreset of size $\tilde{O}(\frac{k^3}{\epsilon^2}) \cdot \tw(G)$.\footnote{Throughout, we use $\tilde{O}(f)$ to denote $O(f\cdot  \operatorname*{polylog}(f))$.}
	Furthermore, the coreset can be computed in time $\tilde{O}(|E|)$ with high probability.\footnote{We note that our size bound can be improved to $\tilde{O}(\frac{k^2}{\epsilon^2}) \cdot \tw(G)$, by replacing Lemma~\ref{lemma:fl11} with Theorem 31 of a recent work~\cite{DBLP:journals/siamcomp/FeldmanSS20}.}
\end{theorem}

We complement our coreset construction with a size lower bound,
which is information-theoretic, i.e., regardless of computational power.
\ifeightpage
\else \begin{theorem}[Coreset Size Lower Bound; see Theorems~\ref{thm:tw_kz_lb}]
	\label{thm:lb_tw_intro}
	For every $0 < \epsilon < 1$ and integers $k, t \geq 1$, there exists a graph $G=(V, E)$ with $\tw(G) \leq t$,
	such that every $\epsilon$-coreset for \kMedian on $X = V$ in $G$ has size $\Omega(\frac{k}{\epsilon} \cdot t)$.
\end{theorem}
\fi \ifeightpage
\else This matches the linear dependence on $\tw(G)$ in our coreset construction,
and we show in Corollary~\ref{cor:lb_general} that the same hard instance actually implies for the first time that the $O(\log n)$ factor is optimal for general metrics, which justifies considering restricted graph families.

Note that we require the coreset to use data points only, which is a natural setting in graphs.
However, a ``continuous'' setting, where the coreset could use interpolated points along an edge, also makes sense in many graph families, such as trees or weighted path graphs.
We show that even if the coreset is given this extra power, the coreset size cannot be reduced significantly, even on weighted path graphs that have treewidth (and actually pathwidth) $1$.
\begin{theorem}[Lower Bound for Lines, Continuous Setting; see Theorem~\ref{1d1c}]
	\label{thm:1d1c_intro}
	For every $0<\epsilon<1/24$ and integer $k\geq 1$, there exists a set of data points $V\subseteq \mathbb{R}$, such that every $\epsilon$-coreset for \kMedian of $V$ has size $\Omega(\frac{k}{\sqrt{\epsilon}})$, even if the coreset may use any point in $\mathbb{R}$.
\end{theorem}
This continuous setting for $V\subseteq \R$ is equivalent to the well-studied Euclidean setting in one dimension, in which coresets could use points in the ambient space.
While this lower bound still has a gap of at least $\frac{1}{\sqrt{\epsilon}}$ from the known upper bounds~\cite{DBLP:journals/dcg/Har-PeledK07}, it is in fact the first nontrivial lower bound for the Euclidean setting.
\fi 

\paragraph{Experiments}
We evaluate our coreset on real-world road networks.
Thanks to our new near-linear time algorithm, the coreset construction scales well even on data sets with millions of points.
Our coreset consistently achieves $<5\%$ error using only $1000$ points on various distributions of data points $X$,
and the small size of the coreset results in a
100x-1000x speedup of local search approximation algorithm for graph \kMedian.
When experimenting with our coreset on different data sets $X$,
we observe that coresets of similar size yield similar error,
which confirms our theoretical bounds (for structured graphs)
where the coreset size is independent of the data set.

In fact, our experiments demonstrate that the algorithm performs well
even without knowing the treewidth of the graph $G$.
More precisely, the algorithm can be executed on an arbitrary graph $G$,
and the treewidth parameter is needed only to tune the coreset size.
We do not know the treewidth of the graphs used in the experiments
(we made no attempt to compute it, even approximately).
Our experiments validate the algorithm's effectiveness in practice,
with coreset size much smaller than our worst-case theoretical guarantees.
In fact, it is also plausible that while the graphs have moderate treewidth,
they are actually ``close'' to having an even smaller treewidth.
Another possible explanation is that the algorithm actually works well
on a wider family of graphs than bounded treewidth,
hence it is an interesting open question to analyze our construction
for graphs that are planar or excluding a fixed minor.

\subsection{Technical Contributions}
\label{sec:tech_contrib}

Our coreset construction employs the importance sampling framework proposed
by Feldman and Langberg~\cite{feldman2011unified},
although implemented differently as explained in Remark~\ref{remark:fl11}.
A key observation of the framework is that it suffices to give a \emph{uniform} upper bound on the \emph{shattering dimension} (see Definition~\ref{def:sdim}),
denoted $\SDim_v(M)$,
of the metric $M=(V,d)$ weighted by any point weight $v : V \to \mathbb{R}_+$.
Our main technical contribution is a (uniform) shattering-dimension bound that is \emph{linear} in the treewidth, and this implies the size bound of our coreset.
\begin{theorem}[Shattering Dimension Bound; see Theorem~\ref{thm:treewidth}]
	\label{thm:treewidth_intro}
	For every edge-weighted graph $G=(V, E)$ and every point weight function $v : V \to \mathbb{R}_{+}$, the shortest-path metric $M$ of $G$ satisfies $\SDim_v(M) \leq O(\tw(G))$.
\end{theorem}

The shattering dimension of many important spaces was studied, including for Euclidean spaces~\cite{feldman2011unified} and for doubling spaces~\cite{huang2018epsilon}.
For graphs, the shattering dimension of an $K_r$-minor free graph
(which includes bounded-treewidth graphs)
is known to be $O(r)$~\cite{bousquet2015vc}
for \emph{unit weight} $v \equiv 1$,
see Section~\ref{sec:preliminaries} for details.
However, a general point weight $v : V \to \mathbb{R}_{+}$
introduces a significant technical challenge which is illustrated below.

In our context, the shattering dimension is defined with respect to
the set system of all \emph{$v$-weighted} metric balls,
where every such ball has a center $x \in V$ and a radius $r \geq 0$,
and is defined by
\begin{equation}
  \label{eqn:Bv}
  B_v(x, r) := \{ y\in V : v(y) \cdot d(x, y) \leq r \}.
\end{equation}
Roughly speaking, a bounded shattering dimension means that
for every subset $H\subseteq V$, the number of ways this $H$ is intersected by $v$-weighted metric balls is at most $\poly(\card{H})$.
The main technical difficulty is that an arbitrary weight $v$
can completely break the ``continuity'' of the space,
which can be illustrated even in one-dimensional line $V=\R$
(and analogously in a simple path graph on $V=\{0, 1, \ldots, n\}$),
where under unit weight $v\equiv 1$, every ball is a \emph{contiguous} interval,
but under a general weight $v$ an arbitrary subset of points could form a ball;
indeed, for a center $x = 0$ and radius $r = 1$,
every point $y\ge1$ can be made inside or outside of the ball $B_v(x, r)$
by setting $v(y) = \frac{1}{2y}$ or $v(y) = \frac{2}{y}$.

Our main technical contribution is to analyze the shattering dimension
with general weight functions,
which we outline now briefly
(see Section~\ref{sec:upper_bound} for a more formal overview).
We start by showing a slightly modified balanced-separator theorem
for bounded-treewidth graphs (Lemma~\ref{lemma:treewidth_key}),
through which the problem of bounding the shattering dimension
is reduced to bounding the ``complexity'' of shortest paths
that cross one of a few vertex separators, each of size $O(\tw(G))$.
An important observation is that, if $S\subset V$ is a vertex separator
and $x,y\in V$ belong to different components after removing $S$,
then every path connecting $x$ to $y$ must \emph{cross} $S$,
and hence
\begin{equation*}
	d(x, y) = \minn{ d(x, s_i) + d(s_i, y) : s_i\in S}.
\end{equation*}
If we fix $x\in V$ and consider all $y\in V$,
then we can think of each $d(s_i,y)$ as a real variable $z_i\in \R$,
so instead of varying over all $y\in V$, which depends on the graph structure,
we can vary over $\card{S}$ real variables, and each $d(x,\cdot)$ is the minimum of $\card{S}$ linear (actually affine) functions, or in short a min-linear function.
Finally, we consider different $x\in V$ with the same separator $S$,
and hence the same $\card{S}$ real variables,
and we bound the ``complexity'' of these min-linear functions
by relating it to the \emph{arrangement number} of hyperplanes,
which is a well-studied concept in computational geometry.
We believe our techniques may be useful for more general graph families,
such as minor-free graphs.

\subsection{Related Work}
Approximation algorithms have been extensively studied
for \kMedian in graph metrics, and here we only mention a small selection of results.
In general graphs (which is equivalent to general metrics),
it is NP-hard to approximate \kMedian within $1 + \frac{2}{e}$ factor~\cite{DBLP:conf/stoc/JainMS02},
and the
state-of-art is a $2.675$-approximation~\cite{DBLP:journals/talg/ByrkaPRST17}.
For planar graphs and more generally graphs excluding a fixed minor,
a PTAS for \kMedian was obtained in~\cite{DBLP:journals/siamcomp/Cohen-AddadKM19} based on local search, and it has been
improved to be FPT (i.e. the running time is of the form $f(k, \epsilon)\cdot n^{O(1)}$) recently~\cite{DBLP:conf/esa/Cohen-AddadPP19}.
For general graphs, \cite{Thorup05} proposed an $O(1)$-approximation that runs in near-linear time.

Coresets have been studied for many problems in addition to \kMedian, such as PCA~\cite{DBLP:journals/siamcomp/FeldmanSS20} and regression~\cite{DBLP:conf/nips/MaaloufJF19}, but in our context we focus on discussing results for other clustering problems only.
For \kCenter clustering in Euclidean space $\mathbb{R}^{\mathsf{d}}$, an $\eps$-coreset of size $O(\frac{k}{\eps^\mathsf{d}})$ can be constructed in near-linear time~\cite{agarwal2002exact,har2004clustering}.
Recently, coreset for generalized clustering objective receives attention from the research community, for example, \cite{DBLP:conf/icml/BravermanJKW19} obtained simultaneous coreset for \OkM, \cite{DBLP:journals/corr/abs-1812-10854,DBLP:conf/nips/HuangJV19} gave a coresets for $k$-clustering with fairness constraints,
and
\cite{DBLP:conf/nips/MaromF19} presented a coreset for \kMeans clustering on lines in Euclidean spaces where inputs are lines in $\mathbb{R}^{\mathsf{d}}$ while the centers are points.

 	\section{Preliminaries}
\label{sec:preliminaries}
\begin{definition}[Tree Decomposition and Treewidth]
	\label{def:treewidth}
	A tree decomposition of a graph $G=(V, E)$ is a tree $\calT$ with node set $\calV$, such that each node in $\calV$, called a \emph{bag}, is a subset of $V$, and the following conditions hold: 
	\begin{enumerate}
		\item $\bigcup_{S \in \calV}{ S } = V$.
		\item $\forall u \in V$, the nodes of $\calT$ that contain $u$ form a connected component in $\calT$.\label{itm:tree_decomp_connect}
		\item $\forall (u, w) \in E$, $\exists S \in \calV$, such that $\{u, w\} \subseteq S$.
	\end{enumerate}
	The treewidth of a graph $G$, denoted $\tw(G)$, is the smallest integer $t$, such that there exists a tree decomposition with maximum bag size $t + 1$.
\end{definition}
A \emph{nice} tree decomposition is a tree decomposition such that each bag has a degree at most $3$.\footnote{Usually, nice tree decompositions are defined to have additional guarantees, but we only need the bounded degree.}
It is well known that there exists a nice tree decomposition of $G$ with maximum bag size $O(\tw(G))$~\cite{kloks1994treewidth}.

\paragraph{Shattering Dimension}
As mentioned in Section~\ref{sec:intro}, our coreset construction employs the \FL framework~\cite{feldman2011unified}. A key notion in the \FL framework is the \emph{shattering dimension} of a metric space with respect to a point weight function.
\begin{definition}[Shattering Dimension]
	\label{def:sdim}
	Given a point weight function $v : V \to \mathbb{R}_{+}$,
	the shattering dimension of $M=(V, d)$ with respect to $v$,
	denoted as $\SDim_v(M)$,
	is the smallest integer $t$,
	such that for every $H \subseteq V$ with $|H| \geq 2$,
	it holds that
	\begin{align*}
	\left| \left\{ H \cap B_v(x, r) : x \in V, r \geq 0 \right\} \right| \leq |H|^t . 
	\end{align*}
\end{definition}
Observe that the left-hand side 
counts the number of ways that $H$ is intersected by all weighted balls,
which were defined in~\eqref{eqn:Bv}. 
We remark that our notion shattering dimension is tightly related to the well-known \emph{VC-dimension} (see for example~\cite{kearns1994introduction}).
In particular, let $\calB_v := \{ B_v(x, r) : x \in V, r \geq 0 \}$
be the collection of all $v$-weighted balls, then the VC-dimension of the set system $(V, \calB_v)$ is within a logarithmic factor to the $\SDim_v(M)$.
It was shown in~\cite{bousquet2015vc} that the VC-dimension of a $K_r$-minor free graph with unit weights $v \equiv 1$ is at most $O(r)$,
which immediately implies an $O(r)$ bound also for the shattering dimension (under unit weight $v \equiv 1$).

	\section{Coresets for \kMedian in Graph Metrics}
\label{sec:upper_bound}
In this section, we present a near-linear time construction for $\epsilon$-coreset for \kMedian in graph metrics, whose size is linear in the treewidth. This is formally stated in the following theorem.

\begin{theorem}[Coreset for Graph \kMedian]
	\label{thm:coreset}
	For every edge-weighted graph $G=(V, E)$, $0 < \epsilon, \delta < 1$, and integer $k \geq 1$,
	\kMedian of every data set $X \subseteq V$
        (with respect to the shortest path metric of $G$)
        admits an $\epsilon$-coreset of size $\tilde{O}\left( \frac{k^2}{\epsilon^{2}} \cdot (k\cdot \tw(G)+\log (1/\delta))\right)$.
        Furthermore, it can be computed in time $\tilde{O}(|E|)$ with success probability $1 - \delta$.
\end{theorem}

Our construction is based on the \FL framework~\cite{feldman2011unified},
in which the coreset is constructed using \emph{importance sampling}.
While this framework is quite general, their implementation is tailored
to Euclidean spaces and is less suitable for graphs metrics. 
In addition, their algorithm runs in $\tilde{O}(kn)$ time
assuming access to pairwise distances, 
which is efficient in Euclidean spaces but rather expensive in graphs.

We give an efficient implementation of the \FL framework in graphs, 
and also provide an alternative analysis that is not Euclidean-specific.
A similar strategy was previously employed for constructing coresets
in doubling spaces~\cite{huang2018epsilon}, 
but that implementation is not applicable here because of the same efficiency issue (i.e., it requires oracle access to distances).
We present our implementation and analysis of the framework below,
and then put it all together to prove Theorem~\ref{thm:coreset}. 

\paragraph{Importance Sampling}
At a high level, the importance sampling method consists of two steps.
\begin{enumerate}
\item For each data point $x \in X$, compute an importance $\sigma_x \in \mathbb{R}_+$.
\item Form a coreset by drawing $N$ (to be determined later)
  independent samples from $X$, 
  where each sample picks every $x \in X$ 
  with probability proportional to $\sigma_x$,
  i.e., $p_x := \frac{\sigma_x}{ \sum_{x'\in X}{\sigma_{x'}} }$,
  and assigns it weight $\frac{1}{p_x}$. 
\end{enumerate}

To implement the algorithm, we need to define $\sigma_x$ and $N$. 
Following the \FL framework, 
each importance $\sigma_x$ is an \emph{upper bound} on the \emph{sensitivity} \begin{equation*}
	\sigma^\star_x := \max_{ C \subseteq V, |C| = k}{ \frac{ d(x, C) }{ \cost(X, C) } },
\end{equation*}
which was introduced in~\cite{DBLP:conf/soda/LangbergS10}
and represents the maximum possible contribution of $x$ to the objective
over all center sets $C$. 

Let the \emph{total importance} be $\sigma_X := \sum_{x \in X}{\sigma_x}$. 
Our key tool is the following bound on coreset size $N$
in terms of $\sigma_X$ and a uniform upper bound on $\SDim_v(M)$.
It follows by combining their Theorem 4.1 in~\cite{feldman2011unified} 
with a PAC sampling bound from~\cite{vapnik2015uniform}. 

\begin{lemma}[Theorem 4.1 of~\cite{feldman2011unified}
  together with~\cite{vapnik2015uniform}]
  \label{lemma:fl11}
  Let $\SDim_{\max} := \max_{v : V \to \mathbb{R}_+}{\SDim_v(M)}$.
  Then for
  \begin{equation*}
    N = \tilde{O}\Big( \big(\frac{\sigma_X}{\epsilon}\big)^2
            \big( k \cdot \SDim_{\max} + \log{\frac{1}{\delta}} \big)
    \Big), 
  \end{equation*}
  the importance sampling procedure returns an $\epsilon$-coreset
  with probability at least $1 - \delta$. 
\end{lemma}

\paragraph{Computing $\sigma_x$}
An efficient algorithm to compute $\sigma_x$
was presented in~\cite{varadarajan2012sensitivity},
assuming that an $O(1)$-approximation to \kMedian is given. 
Furthermore, an $O(k)$ bound on the total importance $\sigma_X$ was shown.

\begin{lemma}[\cite{varadarajan2012sensitivity}]
	\label{lemma:vx12}
	Suppose $C^\star$ is a $\rho$-approximate solution to the \kMedian instance.
	Let $\sigma_x := \rho \cdot \big(\frac{d(x, C^\star)}{\cost(X, C^\star)} + \frac{1}{ |C^\star(x)| } \big) $,
	where $C^\star(x) \subseteq X$ is the cluster of $C^\star$ that contains $x$.
	Then $\sigma_X = O(\rho k)$ and
        $$\forall x \in X, \quad \sigma_x \geq \Omega(\sigma^\star_x).$$ 
\end{lemma}
Thus, to construct the coreset in near-linear time,
we need to compute an $O(1)$-approximation $C^\star$ fast, 
for which we use the following result of~\cite{Thorup05}.
\begin{lemma}
	\label{lemma:thorup}
	There is an algorithm that, given as input
        a weighted undirected graph $G=(V, E)$ and data set $X \subseteq V$,
        computes an $O(1)$-approximate solution for graph \kMedian
        in time $\tilde{O}(|E|)$ with probability $1 - o(1)$. 
\end{lemma}

Finally, we need a uniform shattering-dimension bound 
(with respect to treewidth).
Such a bound, stated next, is our main technical contribution
and its proof is presented in Section~\ref{sec:proof_treewidth}. 
\begin{theorem}[Shattering Dimension]
	\label{thm:treewidth}
	For every edge-weighted graph $G=(V, E)$ and every point weight function $v : V \to \mathbb{R}_{+}$, the shortest-path metric $M$ of $G$ satisfies $\SDim_v(M) \leq O(\tw(G))$.
\end{theorem}

\begin{remark}
  \label{remark:fl11}
  Our implementation of the framework of \cite{feldman2011unified}
  differs in several respects.
  First, their shattering dimension is defined with respect to 
  \emph{hyperbolic} balls instead of usual metric balls
  (as the underlying set system).
  Second, the choice of $\sigma_x$ and the sampling bound are different. 
  While they achieve an improved coreset size (linear in $k$), 
  their analysis relies on Euclidean-specific properties
  and does not apply in graph metrics. 
\end{remark}

\paragraph{Putting It Together} We are now in position to conclude our main result.

\begin{proof}[Proof of Theorem~\ref{thm:coreset}]
Construct a coreset by the importance sampling procedure,
where the importance $\sigma_x$ is computed using Lemma~\ref{lemma:thorup}.
Then we can apply Lemma~\ref{lemma:vx12} with $\rho=O(1)$
to bound the total importance $\sigma_X=O(k)$. 
Combining this and the shattering dimension from Theorem~\ref{thm:treewidth}, 
we can apply Lemma~\ref{lemma:fl11} with coreset size 
\begin{align*}
  N = \tilde{O}\Big(\frac{k^2}{\epsilon^2} \big(k \cdot \tw(G) + \log{\frac{1}{\delta}} \big) \Big).
\end{align*}
The running time is dominated by computing the importance $\sigma_x$ for all $x\in X$,
which we claim can be computed in time $\tilde{O}(|E|)$
by using Lemmas~\ref{lemma:vx12} and~\ref{lemma:thorup}.
Indeed, first compute $C^\star$ in time $\tilde{O}(|E|)$
using Lemma~\ref{lemma:thorup},
then compute the clustering of $X$ with respect to $C^\star$
and the associated distances $\{ d(x, C^\star) : x \in X \}$
using a \emph{single} Dijkstra execution in time $\tilde{O}(|E|)$ time
(see Observation 1 of~\cite{Thorup05}).
Finally, use this information to compute $\sigma_x$ for all $x\in X$,
and sample according to it, in total time $\tilde{O}(|X|)$.
\end{proof}

\subsection{Bounding the Shattering Dimension}
\label{sec:proof_treewidth}
We give a technical overview before presenting the detailed proof
of Theorem~\ref{thm:treewidth}. 
Recall that the unit-weight case of shattering dimension was already proved in~\cite{bousquet2015vc},
and our focus is when $v : V \to \mathbb{R}_+$ is a general weight function.

The proof starts with
a slightly modified balanced-separator theorem for bounded treewidth graphs (Lemma~\ref{lemma:treewidth_key}), through which the problem of bounding the shattering dimension is reduced to bounding the complexity of \emph{bag-crossing} shortest paths for every bag.
A well-known fact is that every bag $\{s_i, \ldots, s_m\} \subseteq V$
in the tree decomposition is a vertex cut of size $m = O(\tw(G))$,
and this leads to an important observation:
if $x$ and $y$ belong to different components after removing this bag,
then every path connecting $x$ with $y$ \emph{crosses} the bag,
and hence
\begin{equation*}
	d(x, y) = \minn{ d(x, s_i) + d(s_i, y) : i\in[m]}.
\end{equation*}

Now suppose we fix $x \in V$ and let $y$ vary over $V$;
then we can write $d(x, \cdot)$ as
a min-linear function (which means the minimum of $m$ linear functions)
$f_x: \R^m\to\R_+$, whose variables are $z_i = d(s_i, y)$ for $i\in[m]$;
notice that the terms $d(x, s_i)$ are constant with respect to $y$.

This alternative view of distances enables us to bound the complexity of shortest-paths, because the functions $\set{f_x}_x$ all have common variables $\set{z_i=d(s_i, y)}_{i\in[m]}$ in real domain (instead of variables in $V$),
and more importantly,
the domain of these functions has low dimension $m=O(\tw(G))$.
Furthermore, the min-linear description also handles weights because $v(x) \cdot f_x$ is min-linear too.
Finally, in a technical lemma (Lemma~\ref{lemma:min_sum_ordering}), we relate the complexity of a collection of min-linear functions of low dimension to the \emph{arrangement number} of hyperplanes, which is a well-studied quantity in computational geometry.

\ifeightpage
\else 

\begin{theorem}[Restatement of Theorem~\ref{thm:treewidth}]
	For every edge-weighted graph $G=(V, E)$ and every point weight function $v : V \to \mathbb{R}_{+}$, the shortest-path metric $M$ of $G$ satisfies $\SDim_v(M) \leq O(\tw(G))$.
\end{theorem}
\begin{proof}
	Fix a point weight $v : V \to \mathbb{R}_+$. We bound the shattering dimension by verifying the definition (see Definition~\ref{def:sdim}).
Fix a subset of points $H \subseteq V$ with $|H| \geq 2$.
	By Definition~\ref{def:sdim}, we need to show
	\begin{align*}
	\left|\left\{ H \cap B_v(x, r) : x \in V, r \geq 0 \right\}\right| \leq |H|^{O(\tw(G))}.
	\end{align*}
	We interpret this as a counting problem, in which we count the number of distinct subsets $H \cap B_v(x, r)$ over two variables $x$ and $r$.
	To make the counting easier, our first step is to ``remove'' the variable $r$, so that we could deal with the center $x$ only.
	
	\paragraph{Relating to Permutations}
	For $x \in V$, let $\pi_x$ be the permutation of $H$
	such that points $y \in H$ are ordered by $d(x, y) \cdot v(y)$ (in non-increasing order) and ties are broken consistently.
	Since $H \cap B_v(x, r)$ corresponds to a prefix of $\pi_x$, and every $\pi_x$ has at most $|H|$ prefixes, we have
\begin{align*}
	\left| \left\{ H \cap B_v(x, r) : x \in V, r\geq 0 \right\} \right| \leq |H| \cdot \left| \{ \pi_x : x \in V \} \right|.
	\end{align*}
	Hence it suffices to show
	\begin{align}
	\left| \left\{ \pi_x : x\in V  \right\} \right| \leq |H|^{O(\tw(G))}, \label{eqn:perm_bound}.
	\end{align}

	Next, we divide the graph (not necessarily a partition) into $\mathrm{poly}(|H|)$ parts using the following structural lemma of bounded treewidth graphs, so that each part is ``simply structured''.
	We prove the following lemma in Section~\ref{sec:proof_treewidth_key}.

	\begin{lemma}[Structural Lemma]
		\label{lemma:treewidth_key}
		Given graph $G(V, E)$, and $H \subseteq V$,
		there exists a collection $\calS \subseteq 2^V$ of subsets of $V$, such that the following holds.
		\begin{enumerate}
			\item $\bigcup_{A \in \calS}{A} = V$.
			\item $|\calS| \leq \mathrm{poly}(|H|)$.
			\item For each $A \in \calS$, either $|A| \leq O(\tw(G))$, or i) $|A \cap H| \leq O(\tw(G))$ and ii) there exists $P \subseteq V$ with $|P| \leq O(\tw(G))$ such that there is no edge in $E$ between $A$ and $V\setminus (A \cup P$).
		\end{enumerate}
	\end{lemma}
Let $\calS$ be the collection of subsets asserted by Lemma~\ref{lemma:treewidth_key}.
Since $\bigcup_{A \in \calS}{A} = V$ and $|\calS| \leq \mathrm{poly}(|H|)$,
	it suffices to count the number of permutations for each part. Formally, it suffices to show that
	\begin{align}
          \forall A \in \calS,
          \qquad
	\left|\left\{ \pi_x : x \in A \right\}\right| \leq |H|^{O(\tw(G))}.
	\label{eqn:A}
	\end{align}

	\paragraph{Counting Permutations for Each $A \in \calS$}
	The easy case is when $|A| \leq O(\tw(G))$:
	\begin{align*}
	\left|\left\{ \pi_x : x \in A \right\}\right| \leq |A| \leq O(\tw(G)) \leq |H|^{O(\tw(G))}.
	\end{align*}
	Then we focus on proving Inequality~(\ref{eqn:A}) for the other case, where
	i) $|A\cap H| \leq O(\tw(G))$ and
	ii) there exists $P \subseteq V$ with $|P| \leq O(\tw(G))$
	such that there is no edge between $A$ and $V \setminus (A \cup P)$, by item 3 of Lemma~\ref{lemma:treewidth_key}.
	
	Now fix such an $A$.
Let $H_A := H \cap A$, and let $Q := P \cup H_A$.
	Write $Q = \{q_1, \ldots, q_m\}$.

	Since there is no edge between $A$ and $V \setminus (A \cup P)$,
	for $x \in A$ and $y \in H$,
	we know that
	\begin{equation*}
	d(x, y) = \min_{q_i \in Q}\{d(x, q_i) + d(q_i, y) \}.
	\end{equation*}
	
	\paragraph{Alternative Representation of \boldmath $d(x, y)$}
	We write $d(x, y)$ in an alternative way.
	If we fix $y \in H$ and vary $x$, then $d(x, y)$ may be represented as a min-linear function in variables $z_i := d(x, q_i)$.
	Specifically, for $y\in H$, define $f_y : \mathbb{R}^{m} \to \mathbb{R}_+$ as
	\begin{align*}
	f_y(z_1, \ldots, z_m) := \min_{i \in [m]}\{ z_i + d(y, q_i)\}.
	\end{align*}
	Note that $d(y, q_i)$ is constant in $f_y$.
	By definition, $f_y(d(x, q_1), \ldots, d(x, q_m)) = d(x, y)$.
	
	We also rewrite $\pi_x$ under this new representation of distances.
	For $a \in \mathbb{R}^{m}$, define $\tau_a$ as a permutation of $H$ that is ordered by $v(y) \cdot f_y(a)$, in the same rule as in $\pi_x$ (i.e. non-decreasing order and ties are broken consistently as in $\pi_x$).
	Then we have
	\begin{align*}
	\pi_x = \tau_{(d(x, q_1), \ldots, d(x, q_m))},
	\end{align*}
	which implies
	\begin{align*}
	\left| \{ \pi_x : x \in A \} \right| \leq \left| \{ \tau_a : a \in \mathbb{R}^m \} \right|.
	\end{align*}
	Thus, it remains to analyze $|\{\tau_a : a \in \mathbb{R}^m \}|$.
	We bound this quantity via the following technical lemma, which describes the complexity of a collection of min-linear functions with bounded dimension.
Its proof appears in Section~\ref{sec:proof_min_linear}.
\begin{lemma}[Complexity of Min-Linear Functions]
		\label{lemma:min_sum_ordering}
		Suppose $f_1, \ldots, f_s$ are $s$ functions such that for every $i\in [s]$,
		\begin{itemize}
			\item $f_i : \mathbb{R}^l \to \mathbb{R}$, and
			\item $f_i(x) = \min_{j \in [l]}\{g_{ij}(x)\}$
			where each $g_{ij} : \mathbb{R}^l \to \mathbb{R}$ is a linear function.
		\end{itemize}
		For $x \in \mathbb{R}^l$, let $\sigma_x$ be the permutation of $[s]$
		such that $i\in [s]$ is ordered by $f_i(x)$ (in non-increasing order), and ties are broken consistently.
		Then $|\{ \sigma_x : x \in \mathbb{R}^l \}| \leq O(sl)^{O(l)}$.
	\end{lemma}
	
	Applying Lemma~\ref{lemma:min_sum_ordering} with
	$s = |H|$, $l = m$ and the collection of min-linear functions $\{ v(y) \cdot f_y : y \in H\}$, we conclude that
\begin{align*}
	\left|\left\{ \tau_a : a \in \mathbb{R}^m \right\}\right| \leq |H|^{O(m)} \leq |H|^{O(\tw(G))}
	\end{align*}

	where the last inequality follows from $|H_A| \leq O(\tw(G))$ and $|P| \leq O(\tw(G))$ and $m=|Q|=|H_A\cup P|\leq |H_A|+|P|\leq O(\tw(G))$.
\end{proof}

\subsection{Proof of the Structural Lemma}
\label{sec:proof_treewidth_key}
\begin{lemma}[Restatement of Lemma~\ref{lemma:treewidth_key}]
	Given graph $G(V, E)$, and $H \subseteq V$,
	there exists a collection $\calS \subseteq 2^V$ of subsets of $V$, such that the following holds.
	\begin{enumerate}
		\item $\bigcup_{A \in \calS}{A} = V$.
		\item $|\calS| \leq \mathrm{poly}(|H|)$.
		\item For each $A \in \calS$, either $|A| \leq O(\tw(G))$, or i) $|A \cap H| \leq O(\tw(G))$ and ii) there exists $P \subseteq V$ with $|P| \leq O(\tw(G))$ such that there is no edge in $E$ between $A$ and $V\setminus (A \cup P)$.
	\end{enumerate}
\end{lemma}

\begin{proof}
	Let $\calT$ be a nice tree decomposition of $G(V, E)$ with maximum bag size $O(\tw(G))$ (see Section~\ref{sec:preliminaries}).
	For a subtree $\calT^\star$ of $\calT$, let $V(\calT^\star)$ be the union of points in all bags of $\calT^\star$.
	For a subset of bags $\calB$ of $\calT$,
	\begin{itemize}
		\item Define $\calT_{\setminus \calB}$ as the set of subtrees of $\calT$ resulted by removing all bags in $\calB$ from $\calT$;
		\item for $\calT^\star \in \calT_{\setminus\calB}$, define $\partial_{\calB}(\calT^\star) \subseteq \calB$ as the subset of bags in $\calB$ via which $\calT^\star$ connects to bags outside of $\calT^\star$;
		\item for $\calT^\star \in \calT_{\setminus\calB}$, define $V_{\setminus \calB}(\calT^\star) := V(\calT^\star) \setminus \bigcup_{S \in \partial_\calB(\calT^\star)}{S}$.
	\end{itemize}
	
	Given a nice tree decomposition $\calT$, we have the following theorem for constructing balanced separators, which will be useful for defining $\mathcal{S}$.

	\begin{theorem}[Balanced Separator of A Tree Decomposition~\cite{robertson1986graph}]
		\label{lemma:seprator}
		Suppose $\calT^\star$ is a subtree of $\calT$ and $w : V \to \{0, 1\}$ is a point weight function.
There exists a bag $S$ in $\calT^\star$, such that any subtree $\calT' \in \calT^\star_{\setminus S}$ satisfies $w(V_{\setminus S}(\calT')) \leq \frac{2}{3} w(V(\calT^\star))$.
	\end{theorem}
	The first step is to construct a subset of bags that satisfy the following nice structural properties.

	\begin{lemma}
		\label{lemma:decomposition}
		There exists a subset of bags $\calB = \calB_H$ of $\calT$, such that the following holds.
		\begin{enumerate}
			\item \label{itm:num_bag}
			$|\calB| = \mathrm{poly}(|H|)$.
			\item \label{itm:bag_size}
			For every $S \in \calB$, $|S| = O(\tw(G))$.
			\item \label{itm:num_boundary}
			For every $\calT^\star \in \calT_{\setminus\calB}$,
			$|\partial_\calB(\calT^\star)| \leq 2$.
			\item \label{itm:internal}
			For every $\calT^\star \in \calT_{\setminus \calB}$,
			$|V_{\setminus \calB}(\calT^\star) \cap H| = O(\tw(G))$.
		\end{enumerate}
	\end{lemma}
	\begin{proof}
		The proof strategy is to start with a set of bags $\calB_1$ such that items~\ref{itm:num_bag},~\ref{itm:bag_size} and~\ref{itm:internal} hold.
		Then for each $\calT^\star \in \calT_{\setminus\calB}$,
		we further ``divide'' it by a few more bags,
		and the newly added bags $\calB_2$ combined with $\calB_1$
		would satisfy all items.
		
		To construct $\calB_1$, we apply Theorem~\ref{lemma:seprator} which constructs balanced separators.
The first step of our argument is no different from constructing a balanced separator decomposition, expect that we need explicitly that each separator is a bag of $\calT$.
		We describe our balanced separator decomposition in Algorithm~\ref{alg:decomp} which makes use of Theorem~\ref{lemma:seprator}.
		
		\begin{algorithm}[t]
			\caption{$\Decomp(\calT^\star, w)$}
		      \label{alg:decomp}
        \begin{algorithmic}[1]
			\INPUT subtree $\calT^\star$ of $\calT$, point weight $w : V \to \{0, 1\}$
		    \OUTPUT set of bags $\calB$ 
			\IF{$w(V(\calT^\star)) \leq O(\tw(G))$}
				\STATE return $\calB \gets  \emptyset$ \label{line:terminate}\;
				\\
				\slash\slash \ stop the decomposition if the subtree has small enough weight
			\ELSE
		      \STATE apply Theorem~\ref{lemma:seprator} on $(\calT^\star, w)$, let $S$ be the asserted bag, and write $\calT^\star_{\setminus S} \gets \{\calT^\star_1, \ldots, \calT^\star_l\}$ \;
		\STATE		define point weight $w' : V \to \{0, 1\}$, such that $w'(u)  = 0$ if $u \in S$ and $w'(u) = w(u)$ otherwise\label{line:zero_out} \;
		\\
				\slash\slash \  zero out weights in $S$, so $w'(V(\calT^\star_i)) = w(V_{\setminus S}(\calT^\star_i))$ 
		\STATE		for $i \in [l]$, let $\calB_i' \gets \Decomp(\calT^\star_i, w')$ \;
			\\
			\slash\slash \  recursively decompose the subtrees using updated weight $w'$ (where weights of $S$ are removed) 
		\STATE		return $\calB \gets (\bigcup_{i \in [l]}{\calB_i'}) \cup S$ \;
		     \ENDIF
        \end{algorithmic}
		\end{algorithm}

Define $w: V \to \{0, 1\}$ as $w(u) = 1$ if $u \in H$ and $w(u) = 0$ otherwise.
		Call Algorithm~\ref{alg:decomp} with $(\calT, w)$, and denote the resulted bags as $\calB_1$, i.e. $\calB_1 := \Decomp(\calT, w)$.
		
		\paragraph{Analyzing $\calB_1$}
		We show $\calB_1$ satisfies Items~\ref{itm:num_bag},~\ref{itm:bag_size} and~\ref{itm:internal}.
		\begin{itemize}
			\item Item~\ref{itm:bag_size} is immediate since $\calB_1$ is a set of bags, and the width of the tree decomposition is $O(\tw(G))$.
			\item Since $\calT$ is a nice tree decomposition, each node of it has degree at most $3$. so each recursive invocation of Algorithm~\ref{alg:decomp}
			creates at most $3$ new subtrees,
			and each subtree has its weight decreased by $\frac{1}{3}$ (by Theorem~\ref{lemma:seprator}).
			Moreover, the initial weight is $|H|$, and the recursive calls terminate when the weight is $O(\tw(G))$ (see Line~\ref{line:terminate}),
			we conclude $|\calB_1| = O(3^{\log_{\frac{3}{2}}(|H|)})=O(|H|^{2.71})=\mathrm{poly}(|H|)$, which is item~\ref{itm:num_bag}.
			\item Because of the observation in the comment of Line~\ref{line:zero_out},
			$w(V(T^\star))$ in Line~\ref{line:terminate} is exactly $V_{\setminus \calB_1}(\calT^\star) \cap H$,
			which implies item~\ref{itm:internal}.
		\end{itemize}
		
		We further modify $\calB_1$ so that item~\ref{itm:num_boundary} is satisfied.
		The modification procedure is listed in Algorithm~\ref{alg:bound_reduct}.
		Roughly speaking, we check each subtree $\calT^\star \in \calT_{\setminus \calB_1}$, and if it violates item~\ref{itm:num_boundary},
		we add more bags \emph{inside} $\calT^\star$, i.e. $\calB^\star$ in line~\ref{line:b_star}, so that $|\partial_{\calB_1 \cup \calB^\star}(\calT^\star)| \leq 2$.
		This modification may be viewed as a refinement for the decomposition defined in Algorithm~\ref{alg:decomp}.

\begin{algorithm}[t]
			\caption{$\BoundReduct(\calT, \calB_1)$}
			\label{alg:bound_reduct}
        \begin{algorithmic}[1]
			\INPUT tree decomposition $\calT$, subset of bags $\calB_1$ 
		    \OUTPUT set of bags $\calB_2$ 
			\STATE initialize $\calB_2 \gets \emptyset$ \;
			\FOR{$\calT^\star \in \calT_{\setminus \calB_1}$}
			
				\IF {$|\partial_{\calB_1}(\calT^\star)| > 2$}
				
					\STATE let $\widehat{\calT^\star}$ be the subtree of $\calT$
					formed by including bags in $\partial_{\calB_1}(\calT^\star)$ and their connecting edges to $\calT^\star$ \;
				\STATE let $\widetilde{\calT^\star}$ be the minimal subtree of $\widehat{\calT^\star}$ that contains all bags in $\partial_{\calB_1}(\calT^\star)$ \;
					\STATE let $\calB^\star$ be the set of bags in $\widetilde{\calT^\star}$ with degree at least $3$\label{line:b_star} \;
					\STATE update $\calB_2 \gets \calB_2 \cup \calB^\star$ \;
				\ENDIF
			\ENDFOR
			\STATE return $\calB_2$ \;
            \end{algorithmic}
\end{algorithm}
		
		Call Algorithm~\ref{alg:bound_reduct} with $(\calT, \calB_1)$,
		and let $\calB_2 := \BoundReduct(\calT, \calB_1)$. We formally analyze $\calB := \calB_1 \cup \calB_2$ as follows.
		
		\paragraph{Analyzing $\calB := \calB_1\cup \calB_2$}
		Item~\ref{itm:bag_size} follows immediately since both $\calB_1$ and $\calB_2$ are sets of bags.
		Now consider an iteration of Algorithm~\ref{alg:bound_reduct} on $\calT^\star \in \calT_{\setminus \calB_1}$.
		By the definition of $\calB^\star$,
		we know that $\widetilde{\calT^\star}_{\setminus (\calB_1 \cup \calB^\star)}$ contains paths only (each having two boundary bags).
		Hence each subtree $\calT' \in \widehat{T^\star}_{\setminus (\calB_1 \cup \calB^\star)}$
		satisfies $|\partial_{\calB_1 \cup \calB^\star}(\calT')| \leq 2$.
		Since Algorithm~\ref{alg:bound_reduct} runs in a tree-by-tree basis,
		we conclude item~\ref{itm:num_boundary}.
		
		Still consider one iteration of Algorithm~\ref{alg:bound_reduct}.
		By using item~\ref{itm:tree_decomp_connect} of the definition of the tree decomposition,
		we have that for every $\calT' \in \widehat{\calT^\star}_{\setminus(\calB_1 \cup \calB^\star)}$,
		$V_{\setminus (\calB_1 \cup \calB^\star)}(\calT') \subseteq
		V_{\setminus \calB_1}(\calT^\star)$.
		Combining this with the fact that $\calB_1$ satisfies item~\ref{itm:internal} (as shown above), we conclude that $\calB$ also satisfies item~\ref{itm:internal}.
		Finally, by the fact that the number of nodes of degree at least $3$
		is at most the number of leaves,
		we conclude that $|\calB^\star| \leq |\partial_{\calB_1}(\calT^\star)|$.
		Then
		\begin{align*}
		|\calB_2|
		\leq \sum_{\calT^\star \in \calT_{\setminus \calB_1}}{ |\calB^\star| }
		\leq \sum_{\calT^\star \in \calT_{\setminus \calB_1}}{ |\partial_{\calB_1}(\calT^\star)| }
		\leq O(1) \cdot |\calB_1|,
		\end{align*}
		where the last inequality is by the degree constraint of the nice tree decomposition.
		Therefore, $|\calB| \leq |\calB_1| + |\calB_2| \leq \mathrm{poly}(|H|)$,
		which concludes item~\ref{itm:num_bag}.
		This finishes the proof of Lemma~\ref{lemma:decomposition}.
	\end{proof}
	
	Suppose $\calB$ is the set asserted by Lemma~\ref{lemma:decomposition}.
	Let $\calS := \calB \cup \{ V_{\setminus\calB}(\calT^\star) : \calT^\star \in \calT_{\setminus \calB} \}$.
	It is immediate that $\bigcup_{A \in \calS}{A} = V$.
	By item~\ref{itm:num_bag} of Lemma~\ref{lemma:decomposition}
	and the degree constraint of the nice tree decomposition $\calT$,
	$| \calT_{\setminus \calB} | = \mathrm{poly}(|H|)$.
	Hence, $|\calS| \leq \mathrm{poly}(|H|)$.
	
	By item~\ref{itm:bag_size} of Lemma~\ref{lemma:decomposition}, we know for every $A \in \calB$, $|A| \leq O(\tw(G))$.
	Now consider $\calT^\star \in \calT_{\setminus \calB}$,
	and let $A := V_{\setminus \calB}(\calT^\star) \in \calS$.
	By item~\ref{itm:internal} of Lemma~\ref{lemma:decomposition}, $|A \cap H| \leq O(\tw(G))$.
	Therefore, we only need to show there exists $P \subseteq V$ such that
	$|P| \leq O(\tw(G))$ and there is no edge between $A$ and $V\setminus (A \cup P)$.
We have the following fact for a tree decomposition.
\begin{fact}[A Bag is A Vertex Cut]
		\label{fact:node_cut}
		Suppose $\calT^\star$ is a subtree of the tree decomposition $\calT$,
		and $S$ is a bag in $\calT^\star$.
Then
		\begin{itemize}
			\item For $1 \leq i < j \leq l$,
			$V_{\setminus S}(\calT^\star_i) \cap V_{\setminus S}(\calT^\star_j) = \emptyset$.
			\item There is no edge between $V_{\setminus S}(\calT^\star_i)$
			and $V_{\setminus S}(\calT^\star_j)$ for $1 \leq i < j \leq l$.
			In other words, $S$ is a vertex cut for $V_{\setminus S}(\calT^\star_i)$ for all $i \in [l]$.
		\end{itemize}
	\end{fact}
	
	Define $P := \bigcup_{S \in \partial_{\calB}(\calT^\star)}{S}$.
	By Fact~\ref{fact:node_cut} and item~\ref{itm:num_boundary} of Lemma~\ref{lemma:decomposition}, we know that $|P| \leq O(\tw(G))$, and there is no edge between $A$ and $V \setminus (A \cup P$).
	This finishes the proof of Lemma~\ref{lemma:treewidth_key}.
\end{proof}

\subsection{Complexity of Min-linear Functions}
\label{sec:proof_min_linear}
\begin{lemma}[Restatement of Lemma~\ref{lemma:min_sum_ordering}]
	Suppose we have $s$ functions $f_1, \ldots, f_s$ such that for every $i\in [s]$,
	\begin{itemize}
		\item $f_i : \mathbb{R}^l \to \mathbb{R}$, and
		\item $f_i(x) = \min_{j \in [l]}\{g_{ij}(x)\}$
		where $g_{ij} : \mathbb{R}^l \to \mathbb{R}$ are linear functions.
	\end{itemize}
	For $x \in \mathbb{R}^l$, let $\sigma_x$ be the permutation of $[s]$
	such that $i\in [s]$ is ordered by $f_i(x)$ (in non-increasing order), and ties are broken consistently.
	Then $|\{ \sigma_x : x \in \mathbb{R}^l \}| \leq O(sl)^{O(l)}$.
\end{lemma}
\begin{proof}
	The proof strategy is to relate the number of permutations to the arrangement number of hyperplanes.
The main tool that we use is the upper bound of the number of arrangements of hyperplanes.
Specifically, as stated in Theorem 2.2 of~\cite{sack1999handbook}, $p$ hyperplanes of dimension $d$ can partition $\mathbb{R}^d$ into $O(p)^d$ regions.
At a high level, we start with ``removing'' the $\min$ in $f_i$'s,
	by partitioning $\mathbb{R}^l$ into \emph{linear} regions
	in which $f_i(x)$'s are simply linear functions.
We bound the number of linear regions using the arrangement bound.
	Since $f_i(x)$'s are linear functions in each linear region,
	we may interpret them as $l$-dimensional hyperplanes.
Then, we bound the number of $\sigma_x$'s that are formed by $s$ hyperplanes of dimension $l$ using the arrangement bound again.
	The lemma is thus concluded by combining the two parts.
	We implement the two steps as follows.
	
	We call $R \subseteq \mathbb{R}^l$ a linear region,
	if $R$ is a maximal region satisfying that
	for all $i \in [s]$, there exists $j_i \in [l]$ such that $f_i(x) = g_{ij_i}(x)$ holds for all $x \in R$.
Observe that for each $i \in [s]$ and $j \in [l]$,
	the set of $x \in \mathbb{R}^l$ such that $f_i(x) = g_{ij}(x)$
	may be represented by the intersection of at most $l$ halfspaces of dimension $l$.
	(For example, when $j = 1$, the set is determined by
	$g_{i1}(x) \leq g_{i2}(x)$ and $g_{i1}(x) \leq g_{i3}(x)$ and $\ldots$ and $g_{i1}(x) \leq g_{il}(x)$.)
	Hence, the boundaries of linear regions must be formed by those intersections.
	Therefore, the number of linear regions is upper bounded by $O(sl)^l$ using the arrangement number bound.
	
	Suppose $R \subseteq \mathbb{R}^l$ is a linear region.
	Then for any $i \in [s]$, $f_i(x)$ ($x \in R$)
	may be interpreted as a $l$-dimensional hyperplane $\calP_i$.
	Hence, any maximal subset $S \subseteq R$ such that $\forall x, y \in S$, $\sigma_x = \sigma_y$,
	is a (convex) region whose boundaries are formed by the intersection of (any two of) the hyperplanes $\calP_i$'s (noting that the intersection is of dimension at most $l$).
	We call such $S$'s invariant regions.
Apply the arrangement number bound again, we can upper bound the number of invariant regions in a linear region by $O(s)^{O(l)}$.

	Note that invariant regions subdivide linear regions and each invariant region introduces exactly one permutation $\sigma_x$.
Therefore, we can upper bound the distinct number of permutations by the total number of invariant regions, i.e.,
	$\left| \left\{ \sigma_x : x \in \mathbb{R}^l \right\} \right|
	\leq O(sl)^l \cdot O(s)^{O(l)} \leq O(sl)^{O(l)}$.
This concludes the lemma.
\end{proof}

\fi

 	\section{Experiments}
We implement our algorithm and evaluate its performance on real-world road networks.
Our implementation generally follows the importance sampling algorithm as in Section~\ref{sec:upper_bound}.
We observe that the running time is dominated by
computing an $O(1)$-approximation for \kMedian (used to assign importance $\sigma_x$),
for which we use Thorup's $\tilde{O}(|E|)$-time algorithm (Lemma~\ref{lemma:thorup}).
However, the straightforward implementation of Thorup's algorithm is very complicated and scales with $k \log^3{n}$ which is already near the size of our data set,
and thus we employ an optimized implementation based on it.

\paragraph{Optimized Implementation}

\begin{algorithm}[t]
    \caption{\textsc{IteratedThorupSampling}}
    \label{alg:iterthor}
    \begin{algorithmic}[1]
    \INPUT edge-weighted graph $G = (V, E)$,
    data set $X \subseteq V$, 
    number of centers $k$, parameters $n$ and $m$ that control the number of iterations
    \OUTPUT bicriteria solution $F \subseteq X$

\STATE let $X_0 \gets X$, and $\forall u \in X_0$ let $w_{X_0}(u) \gets 1$
    \FOR{$i = 1$ to $n$}
        \STATE expand $X_{i-1}$ into a multi-set $X'$, such that each $u \in X'$ has multiplicity $w_{X_{i-1}}(u)$
        \STATE let $F_i \gets \textsc{ThoSampleBest}(G, X', k, m)$
        \STATE let $X_i \gets F_i$ and $\forall u \in X_i$, let
        $$w_{X_i}(u) \gets |\{ v \in X : \operatorname{NN}_{F_i}(v) = u \}|$$
        \slash\slash \ $\operatorname{NN}_{F_{i}}(v)$ is the nearest point in $F_{i}$ from $v$; this step implements the projection of $X$ to $F_i$
    \ENDFOR
    \STATE return $F_n$
\end{algorithmic}
\end{algorithm}

\begin{algorithm}[t]
    \caption{\textsc{ThoSampleBest}}
    \label{alg:thormincost}
    \begin{algorithmic}[1]
    \INPUT
    edge-weighted graph $G = (V, E)$, data set $X \subseteq V$, number of centers $k$, number of iterations $m$ 
    \OUTPUT bicriteria solution $F \subseteq X$ 
    \FOR{$i = 1$ to $m$}
        \STATE let $F_i \gets \textsc{ThoSample}(G, X, k)$
        \\
        \slash \slash \ \textsc{ThoSample} is Algorithm D of~\cite{Thorup05}
    \ENDFOR
    \STATE return $F_i$ such that $i = \arg\min_{1 \leq j \leq m}{\cost(F_i, X)}$
    \end{algorithmic}
\end{algorithm}

Thorup's algorithm starts with an $O(|E|\log{|E|})$-time procedure to find a bicriteria solution $F$ (Algorithm D in~\cite{Thorup05}),
namely, $|F| = O(k \log^2{n})$ such that $\cost(F, X) = O(1) \cdot \text{OPT}$.
Then a modified Jain-Vazirani algorithm~\cite{DBLP:journals/jacm/JainV01} is applied on $F$ to produce the final $O(1)$-approximation in $\tilde{O}(|E|)$ time.
However, the modified Jain-Vazirani algorithm is complicated to implement,
and the hidden polylogarithmic factor in its running time is quite large.
Thus, we replace the Jain-Vazirani algorithm with a simple local search algorithm~\cite{localsearch} to
find an $O(1)$-approximation on $F$.
The performance of the local search relies heavily on $|F|$, but $|F| = O(k \log^2 n)$ is not much smaller than $n$ for our data set.
Therefore, we run the bicriteria approximation \emph{iteratively} to further reduce $|F|$. Specifically, after we obtain $F_i$,
we project $X$ to $F_i$ (i.e., map each $x\in X$ to its nearest point in $F_i$) to form $X_i$,
and run the bicriteria algorithm again on $X_i$ to form $F_{i+1}$.
We use a parameter to control the number of iterations, and we observe 
that $F$ reduces significantly in our data set with only a few iterations.

The procedure for finding $F$ iteratively is described in Algorithm~\ref{alg:iterthor},
which uses Algorithm~\ref{alg:thormincost} as a subroutine.
Algorithm~\ref{alg:thormincost} essentially corresponds to the
above-mentioned Thorup's bicriteria approximation algorithm \textsc{ThoSample} (Algorithm D in~\cite{Thorup05}),
except that we execute it multiple times ($m$ times in Algorithm~\ref{alg:iterthor}) to boost the success probability.
As can be seen in our experiments, the improved implementation scales very well on road networks and achieves high accuracy.

\paragraph{Experimental Setup}
Throughout the experiments the graph $G$ is a road network of New York State extracted from OpenStreetMap~\cite{OpenStreetMap}
and clipped by bounding box to enclose New York City (NYC).
This graph consists of 1 million vertices and 1.2 million edges whose weight are the distances calculated using the Haversine formula between the endpoints.
It is illustrated in Figure~\ref{fig:dataset}.
\begin{figure*}[t]
    \centering
    \begin{subfigure}{0.4\textwidth}
        \includegraphics[width=\textwidth]{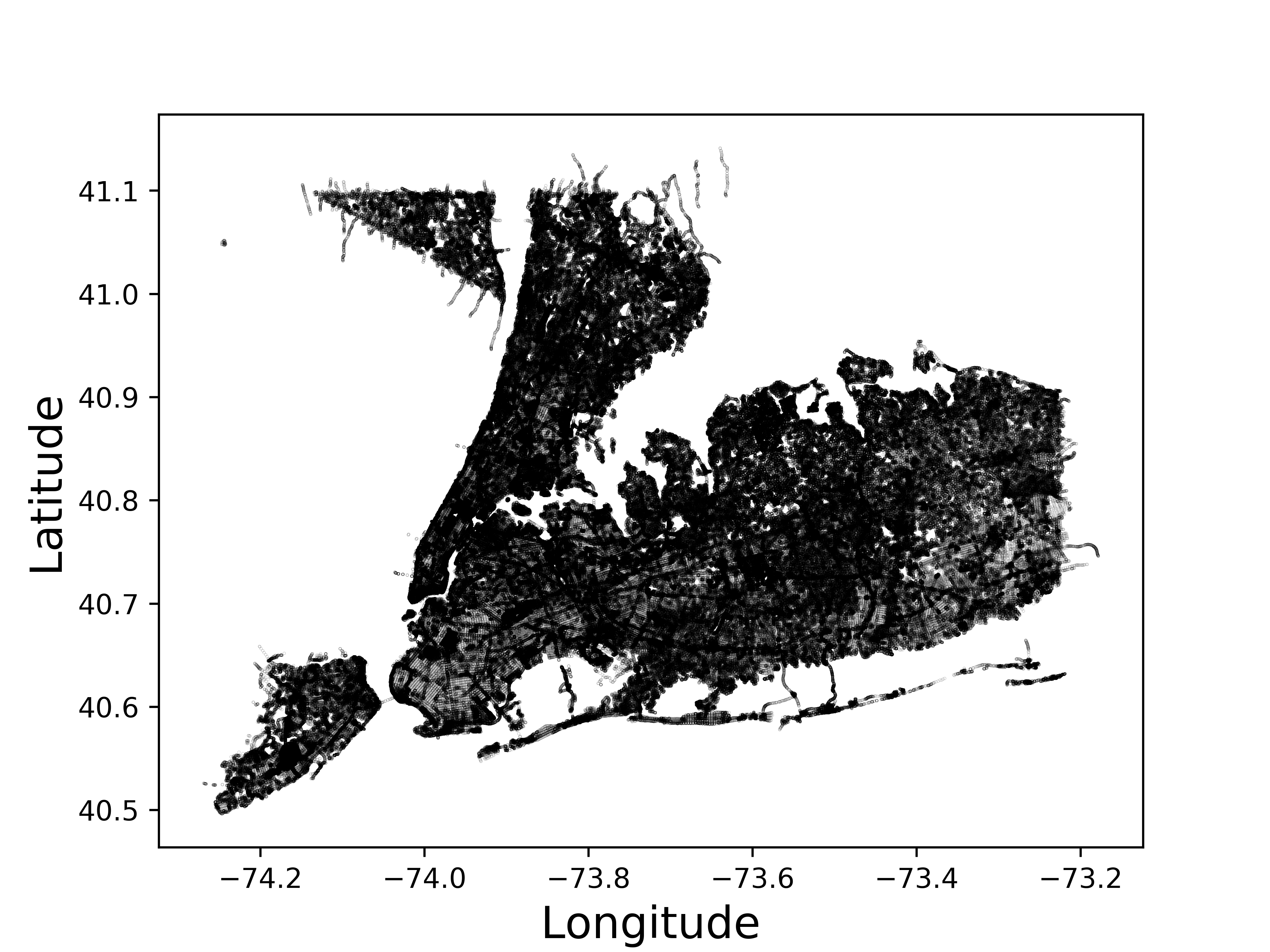}
      \end{subfigure}
      \qquad\qquad
    \begin{subfigure}{0.31\textwidth}
        \includegraphics[width=\textwidth]{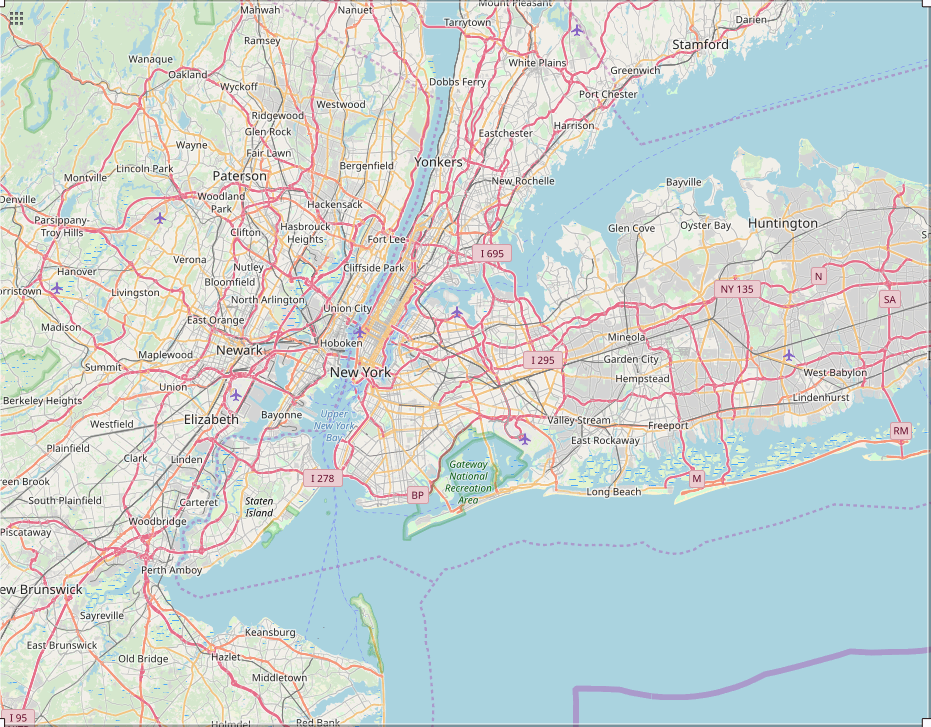}
    \end{subfigure}
    \caption{Illustration of our graph $G$,
      plotting (on left) the vertices according to their geographic coordinates,
      and showing (on right) a map, taken from OpenStreetMap, of the bounding box used to form $G$.
    }
    \label{fig:dataset}
\end{figure*}
Our software is open source and freely available, and implemented in C++17.
All experiments were performed on a Lenovo x3850 X6 system with 4 2.0 GHz Intel E7-4830 CPUs,
each with 14 processor cores with hyperthreading enabled. The system had 1 TB of RAM.

\subsection{Performance of Coresets}
\label{sec:exp_acc}
Our first experiments evaluate how the accuracy of our coresets depends on their size.
Here, the data $X$ may be interpreted as a set of customers to be clustered,
and their distribution could have interesting geographical patterns.
We experiment with $X$ chosen uniformly at random from $V$ (all of NYC),
mostly for completeness as it is less likely in practice,
and denote this scenario as $X_{\mathrm{uni}}$.
We also experiment with a ``concentrated'' scenario where $X$ is highly concentrated in Manhattan but also has much fewer points picked uniformly from other parts of NYC, denoted as $X_{\mathrm{man}}$.
We demonstrate the two types of data sets $X$ in Figure~\ref{fig:x}.

\begin{figure*}[t]
    \centering
    \begin{subfigure}{0.32\textwidth}
        \includegraphics[width=\textwidth]{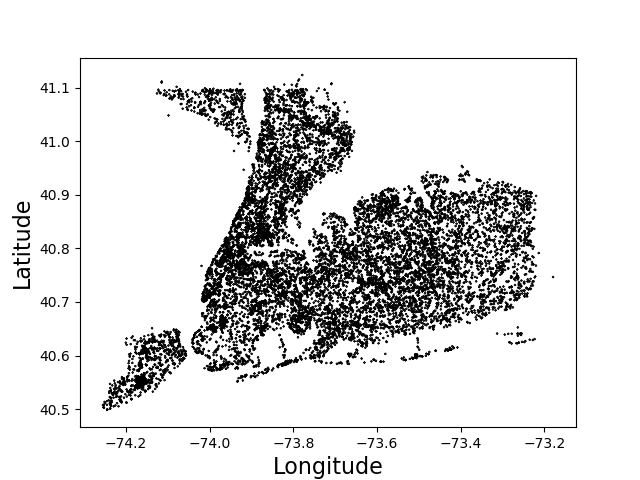}
    \end{subfigure}
    \begin{subfigure}{0.32\textwidth}
        \includegraphics[width=\textwidth]{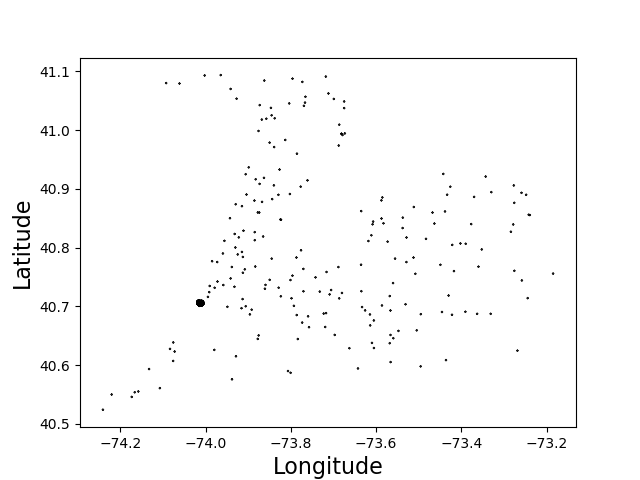}
    \end{subfigure}
    \begin{subfigure}{0.32\textwidth}
        \includegraphics[width=\textwidth]{fig/nyc_data.png}
    \end{subfigure}
    \caption{Illustration of data set $X$ used in the accuracy-vs-size experiment.
      The left plot is a uniform data $X_{\mathrm{uni}}$,
      the middle is $X_{\mathrm{man}}$ that is highly concentrated in Manhattan,
      where in both cases $|X| \approx 14000$,
      and the right plot is all of $V$ which is the full NYC.
    }
    \label{fig:x}
\end{figure*}

\begin{table}[t]
  \caption{Comparison of empirical error of our coreset
    with the baseline of uniform sampling
    when $k = 25$ and varying coreset sizes,
    for both data sets $X_{\mathrm{uni}}$ and $X_{\mathrm{man}}$.
}
\label{tab:acc}
\vskip 0.15in
    \begin{center}
        \begin{small}
            \begin{sc}
                \begin{tabular}{ccccc}
                    \toprule
                    \multirow{2}{*}{size} & \multicolumn{2}{c}{$X_{\mathrm{uni}}$} & \multicolumn{2}{c}{$X_{\mathrm{man}}$}  \\
                    & Ours & Uni. & Ours & Uni.  \\
                    \midrule
25 & 32.1\% & 35.8\% & 32.1\% & 151.6\% \\
                    50 & 26.6\% & 23.0\% & 22.1\% & 90.3\% \\
                    75 & 17.8\% & 23.2\% & 23.2\% & 62.3\% \\
                    100 & 17.2\% & 17.2\% & 15.2\% & 49.9\% \\
500 & 7.72\% & 8.53\% & 8.34\% & 31.7\% \\
1250 & 4.57\% & 5.32\% & 4.87\% & 21.2\% \\
                    2500 & 4.14\% & 4.03\% & 3.29\% & 9.53\% \\
                    3750 & 2.49\% & 3.21\% & 2.89\% & 14.39\% \\
                    6561 & 2.00\% & 2.11\% & 2.38\% & 5.83\% \\
                    13122 & 1.50\% & 1.70\% & 1.53\% & 6.53\% \\
                    19683 & 1.27\% & 1.36\% & 1.39\% & 3.73\% \\
                    \bottomrule
                \end{tabular}
            \end{sc}
        \end{small}
    \end{center}
\vskip -0.1in
\end{table}

We define the \emph{empirical error} of a coreset $D$
and a center set $C \subseteq V$ as
$\operatorname{err}(D, C) := \left| \frac{\cost(D, C)}{\cost(X, C)} - 1 \right|$ (corresponding to $\epsilon$ in the definition of a coreset).
Since by definition a coreset preserves the objective for all center sets,
we evaluate the empirical error by randomly picking $2000$ center sets $C \subseteq V$ from $V$,
and reporting the \emph{maximum} empirical error $\operatorname{err}(D, C)$ over all these $C$.
For the sake of evaluation,
we compare the maximum empirical error of our coreset
with a baseline of a uniform sample, where points are drawn
uniformly at random from $X$ and assigned equal weight (that sums to $|X|$).
To reduce the variance introduced by the randomness in the coreset construction,
we repeat each construction $10$ times
and report the average of their maximum empirical error.

\paragraph{Results}
We report the empirical error of our coresets
and that of the uniform sampling baseline in Table~\ref{tab:acc}.
Our coreset performs consistently well and quite similarly
on the two data sets $X$,
achieving for example $5\%$ error using only about $1000$ points.
Compared to the uniform sampling baseline,
our coreset is $3-5$ times more accurate on the Manhattan-concentrated data $X_{\mathrm{man}}$, and (as expected)
is comparable to the baseline on the uniform data $X_{\mathrm{uni}}$.

In addition, we show the accuracy of our coresets with respect to varying sizes of data sets $X$ in Figure~\ref{fig:acc_other} (left).
We find that coresets of the same size have similar accuracy regardless of $|X|$,
which confirms our theory that the size of the coreset is independent of $|X|$ in structured graphs.
We also verify in Figure~\ref{fig:acc_other} (right)
that a coreset constructed for a target value $k=25$
performs well also as a coreset for fewer centers (various $k' < k$).
While this should not be surprising and follows from the coreset definition,
it is very useful in practice when $k$ is not known in advance,
and a coreset (constructed for large enough $k$)
can be used to experiment and investigate different $k'<k$.

\begin{figure*}[t]
    \centering
    \begin{subfigure}[t]{0.4\textwidth}
        \includegraphics[width=\textwidth]{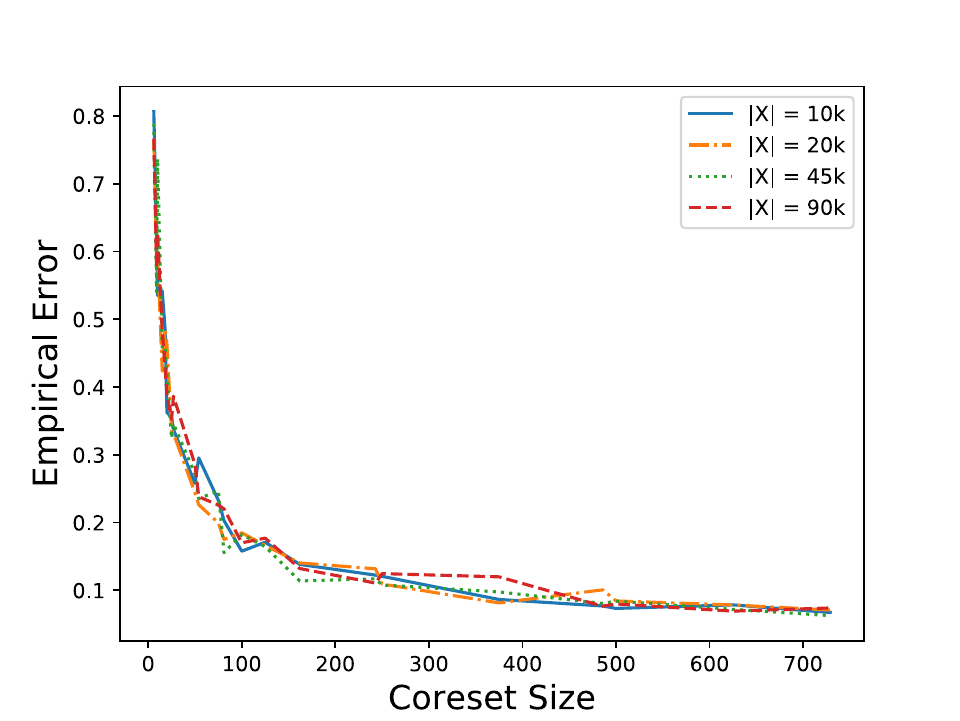}
    \end{subfigure}
    \quad
    \begin{subfigure}[t]{0.4\textwidth}
        \includegraphics[width=\textwidth]{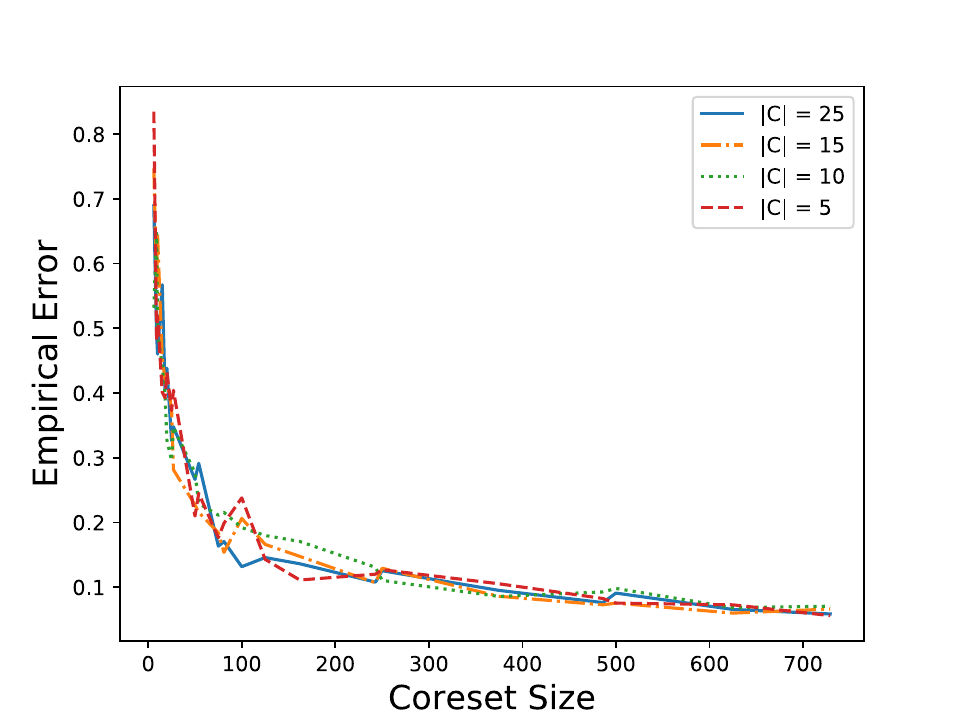}
    \end{subfigure}
    \caption{
The left plot shows the accuracy of coresets $(k = 25)$ on uniform $X$'s with varying sizes. Each line is labeled
    with the size of each respective $X \subseteq V$.
    The right plot shows the accuracy of coresets constructed with $k = 25$ but evaluated with smaller center sets $C$ on the same uniform $X$ with $|X| = 10^4$.
    }
    \label{fig:acc_other}
\end{figure*}

\begin{figure*}[t]
    \centering
    \begin{subfigure}{0.45\textwidth}
        \includegraphics[width=\textwidth]{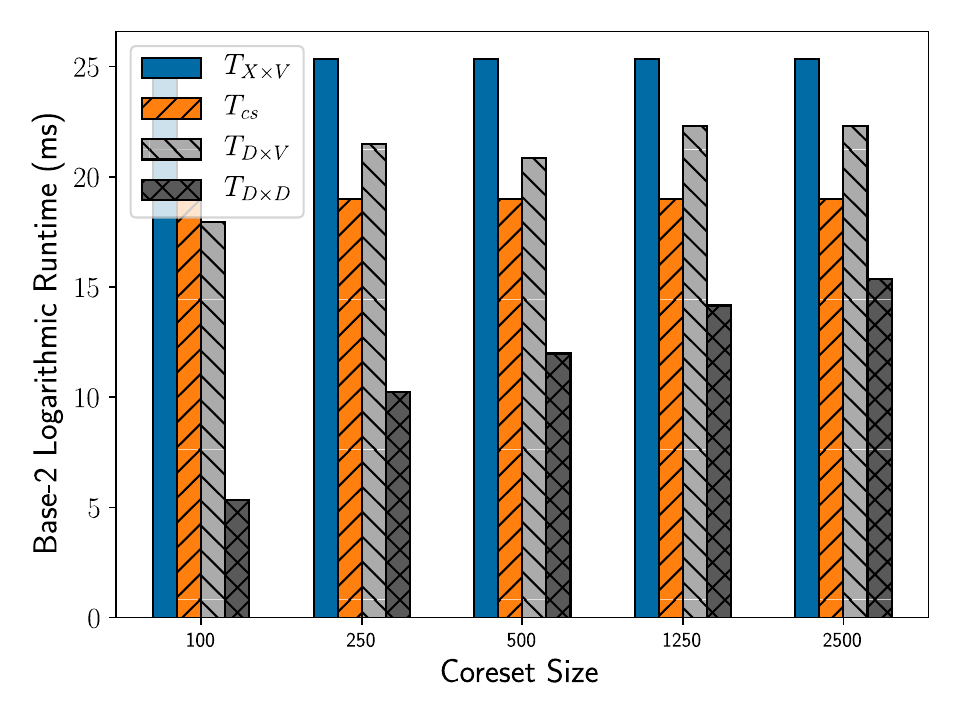}
    \end{subfigure}
    \begin{subfigure}{0.45\textwidth}
        \includegraphics[width=\textwidth]{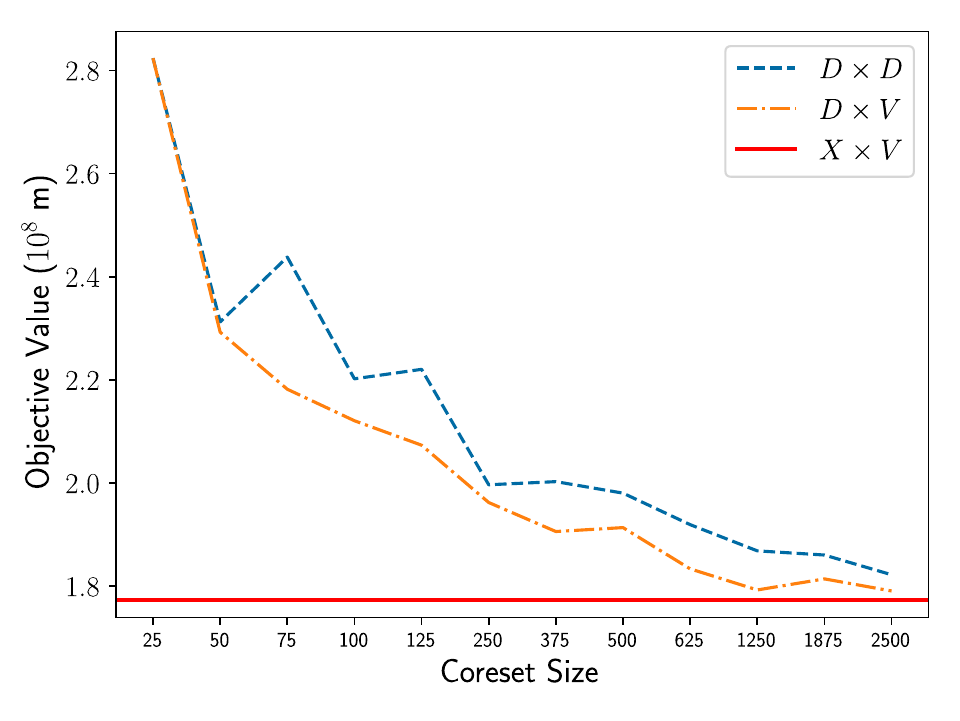}
    \end{subfigure}
    \caption{Performance of local search
    on $X \times V$, $D \times V$, and $D \times D$.
    The running time is shown on the left,
    where the coreset construction time $T_{\text{cs}}$ is separated out
    (so $T_{D \times V}$ and $T_{D \times D}$ do not include $T_{\text{cs}}$).
    The objective values reached are shown on the right.
    Here $k = 25$ and $X = V$ is the whole NYC of size $|X| \approx 10^6$.
    }
    \label{fig:speedup}
\end{figure*}
\subsection{Speedup of Local Search}
An important application of coresets is to speed up existing approximation algorithms.
To this end, we demonstrate the speedup of the local search algorithm of~\cite{localsearch} achieving $5$-approximation for graph \kMedian
by using our coreset.
In particular, we run the local search on top of our coreset $D$ (denoted as $D \times V$), and then compare the accuracy and the overall running time
with those of running the local search on the original data $X$ (denoted as $X \times V$).
Notice that by definition of \kMedian,
the centers always come from $V$, which defines the search space,
and a smaller data set can only affect the time required to evaluate the objective.
This limits the potential speedup of local search,
and therefore we additionally evaluate the running time and accuracy of
local search on $D$ when also the centers come from $D$ (denoted as $D \times D$).

We report separately the running time of the coreset construction, denoted $T_{\text{cs}}$,
and that of the local search on the coreset.
Indeed, as mentioned in Section~\ref{sec:exp_acc},
a coreset $D$ constructed for large $k$ can be used also when clustering for $k' < k$,
and since one can experiment with any clustering algorithm on $D$
(e.g. Jain-Vazirani, local search, etc.),
the coreset construction is one-time effort that may be averaged out
when successive clustering tasks are performed on $D$.

The results are illustrated in Figure~\ref{fig:speedup},
where we find that the coreset construction is very efficient,
about $100$ times faster than local search on $X$,
not to mention that the coreset may be used for successive clustering tasks.
We see that the speedup of local search $D \times V$ is only moderate
(which matches the explanation above),
but the alternative local search on $D \times D$ performs extremely well
--- for example using $|D| \approx 1000$,
it is about $1000$ times faster than the naive local search on $X$,
and it achieves similar objective value (i.e. $5\% - 10\%$ error).
This indicates that local search on $D \times D$
may be a good candidate for practical use.

	\section*{Acknowledgements}
	This research was supported in part by NSF CAREER grant 1652257, ONR Award N00014-18-1-2364, the Lifelong Learning Machines program from DARPA/MTO, Israel Science Foundation grant \#1086/18, and a Minerva Foundation grant.

	\bibliography{references}

\begin{thebibliography}{45}
\providecommand{\natexlab}[1]{#1}
\providecommand{\url}[1]{\texttt{#1}}
\expandafter\ifx\csname urlstyle\endcsname\relax
  \providecommand{\doi}[1]{doi: #1}\else
  \providecommand{\doi}{doi: \begingroup \urlstyle{rm}\Url}\fi

\bibitem[Agarwal \& Procopiuc(2002)Agarwal and Procopiuc]{agarwal2002exact}
Agarwal, P.~K. and Procopiuc, C.~M.
\newblock Exact and approximation algorithms for clustering.
\newblock \emph{Algorithmica}, 33\penalty0 (2):\penalty0 201--226, 2002.

\bibitem[Agarwal et~al.(2004)Agarwal, Har-Peled, and Varadarajan]{AHV04}
Agarwal, P.~K., Har-Peled, S., and Varadarajan, K.~R.
\newblock Approximating extent measures of points.
\newblock \emph{J. ACM}, 51\penalty0 (4):\penalty0 606--635, July 2004.
\newblock ISSN 0004-5411.
\newblock \doi{10.1145/1008731.1008736}.

\bibitem[Arya et~al.(2001)Arya, Garg, Khandekar, Meyerson, Munagala, and
  Pandit]{localsearch}
Arya, V., Garg, N., Khandekar, R., Meyerson, A., Munagala, K., and Pandit, V.
\newblock Local search heuristic for $k$-median and facility location problems.
\newblock In \emph{{STOC}}, pp.\  21--29. {ACM}, 2001.

\bibitem[Balcan et~al.(2013)Balcan, Ehrlich, and Liang]{balcan2013distributed}
Balcan, M.-F.~F., Ehrlich, S., and Liang, Y.
\newblock Distributed $k$-means and $k$-median clustering on general
  topologies.
\newblock In \emph{{NIPS}}, pp.\  1995--2003, 2013.

\bibitem[Bousquet \& Thomass{\'e}(2015)Bousquet and
  Thomass{\'e}]{bousquet2015vc}
Bousquet, N. and Thomass{\'e}, S.
\newblock {VC}-dimension and {Erd{\H{o}}s--P{\'o}sa} property.
\newblock \emph{Discrete Mathematics}, 338\penalty0 (12):\penalty0 2302--2317,
  2015.

\bibitem[Braverman et~al.(2019)Braverman, Jiang, Krauthgamer, and
  Wu]{DBLP:conf/icml/BravermanJKW19}
Braverman, V., Jiang, S.~H., Krauthgamer, R., and Wu, X.
\newblock Coresets for ordered weighted clustering.
\newblock In \emph{{ICML}}, volume~97 of \emph{Proceedings of Machine Learning
  Research}, pp.\  744--753. {PMLR}, 2019.

\bibitem[Byrka et~al.(2017)Byrka, Pensyl, Rybicki, Srinivasan, and
  Trinh]{DBLP:journals/talg/ByrkaPRST17}
Byrka, J., Pensyl, T., Rybicki, B., Srinivasan, A., and Trinh, K.
\newblock An improved approximation for \emph{k}-median and positive
  correlation in budgeted optimization.
\newblock \emph{{ACM} Trans. Algorithms}, 13\penalty0 (2):\penalty0
  23:1--23:31, 2017.

\bibitem[Chen(2009)]{chen2009coresets}
Chen, K.
\newblock On coresets for k-median and k-means clustering in metric and
  euclidean spaces and their applications.
\newblock \emph{SIAM Journal on Computing}, 39\penalty0 (3):\penalty0 923--947,
  2009.

\bibitem[Cohen{-}Addad et~al.(2019{\natexlab{a}})Cohen{-}Addad, Klein, and
  Mathieu]{DBLP:journals/siamcomp/Cohen-AddadKM19}
Cohen{-}Addad, V., Klein, P.~N., and Mathieu, C.
\newblock Local search yields approximation schemes for $k$-means and
  $k$-median in {E}uclidean and minor-free metrics.
\newblock \emph{{SIAM} J. Comput.}, 48\penalty0 (2):\penalty0 644--667,
  2019{\natexlab{a}}.
\newblock \doi{10.1137/17M112717X}.

\bibitem[Cohen{-}Addad et~al.(2019{\natexlab{b}})Cohen{-}Addad, Pilipczuk, and
  Pilipczuk]{DBLP:conf/esa/Cohen-AddadPP19}
Cohen{-}Addad, V., Pilipczuk, M., and Pilipczuk, M.
\newblock Efficient approximation schemes for uniform-cost clustering problems
  in planar graphs.
\newblock In \emph{{ESA}}, volume 144 of \emph{LIPIcs}, pp.\  33:1--33:14.
  Schloss Dagstuhl - Leibniz-Zentrum f{\"{u}}r Informatik, 2019{\natexlab{b}}.

\bibitem[Cui et~al.(2008)Cui, Zhou, Qu, Wong, and Li]{cui2008geometry}
Cui, W., Zhou, H., Qu, H., Wong, P.~C., and Li, X.
\newblock Geometry-based edge clustering for graph visualization.
\newblock \emph{IEEE Transactions on Visualization and Computer Graphics},
  14\penalty0 (6):\penalty0 1277--1284, 2008.

\bibitem[Feldman(2020)]{DBLP:journals/widm/Feldman20}
Feldman, D.
\newblock Core-sets: An updated survey.
\newblock \emph{Wiley Interdiscip. Rev. Data Min. Knowl. Discov.}, 10\penalty0
  (1), 2020.

\bibitem[Feldman \& Langberg(2011)Feldman and Langberg]{feldman2011unified}
Feldman, D. and Langberg, M.
\newblock A unified framework for approximating and clustering data.
\newblock In \emph{43rd Annual ACM Symposium on Theory of computing}, pp.\
  569--578. ACM, 2011.
\newblock Full version at \url{https://arxiv.org/abs/1106.1379}.

\bibitem[Feldman et~al.(2020)Feldman, Schmidt, and
  Sohler]{DBLP:journals/siamcomp/FeldmanSS20}
Feldman, D., Schmidt, M., and Sohler, C.
\newblock Turning big data into tiny data: Constant-size coresets for k-means,
  pca, and projective clustering.
\newblock \emph{{SIAM} J. Comput.}, 49\penalty0 (3):\penalty0 601--657, 2020.

\bibitem[Fichtenberger et~al.(2013)Fichtenberger, Gill{\'{e}}, Schmidt,
  Schwiegelshohn, and Sohler]{DBLP:conf/esa/FichtenbergerGSSS13}
Fichtenberger, H., Gill{\'{e}}, M., Schmidt, M., Schwiegelshohn, C., and
  Sohler, C.
\newblock {BICO:} {BIRCH} meets coresets for $k$-means clustering.
\newblock In \emph{{ESA}}, volume 8125 of \emph{Lecture Notes in Computer
  Science}, pp.\  481--492. Springer, 2013.
\newblock \doi{10.1007/978-3-642-40450-4_41}.

\bibitem[Fortunato(2010)]{fortunato2010community}
Fortunato, S.
\newblock Community detection in graphs.
\newblock \emph{Physics reports}, 486\penalty0 (3-5):\penalty0 75--174, 2010.

\bibitem[Friggstad et~al.(2019)Friggstad, Rezapour, and
  Salavatipour]{DBLP:journals/siamcomp/FriggstadRS19}
Friggstad, Z., Rezapour, M., and Salavatipour, M.~R.
\newblock Local search yields a {PTAS} for $k$-means in doubling metrics.
\newblock \emph{{SIAM} J. Comput.}, 48\penalty0 (2):\penalty0 452--480, 2019.

\bibitem[Har-Peled(2004)]{har2004clustering}
Har-Peled, S.
\newblock Clustering motion.
\newblock \emph{Discrete \& Computational Geometry}, 31\penalty0 (4):\penalty0
  545--565, 2004.

\bibitem[Har{-}Peled \& Kushal(2007)Har{-}Peled and
  Kushal]{DBLP:journals/dcg/Har-PeledK07}
Har{-}Peled, S. and Kushal, A.
\newblock Smaller coresets for $k$-median and $k$-means clustering.
\newblock \emph{Discrete {\&} Computational Geometry}, 37\penalty0
  (1):\penalty0 3--19, 2007.
\newblock \doi{10.1007/s00454-006-1271-x}.

\bibitem[{Har-Peled} \& Mazumdar(2004){Har-Peled} and Mazumdar]{HM04}
{Har-Peled}, S. and Mazumdar, S.
\newblock On coresets for $k$-means and $k$-median clustering.
\newblock In \emph{36th Annual ACM Symposium on Theory of Computing}, pp.\
  291--300, 2004.
\newblock \doi{10.1145/1007352.1007400}.

\bibitem[Herman et~al.(2000)Herman, Melan{\c{c}}on, and
  Marshall]{DBLP:journals/tvcg/HermanMM00}
Herman, I., Melan{\c{c}}on, G., and Marshall, M.~S.
\newblock Graph visualization and navigation in information visualization: {A}
  survey.
\newblock \emph{{IEEE} Trans. Vis. Comput. Graph.}, 6\penalty0 (1):\penalty0
  24--43, 2000.

\bibitem[Huang et~al.(2018)Huang, Jiang, Li, and Wu]{huang2018epsilon}
Huang, L., Jiang, S., Li, J., and Wu, X.
\newblock Epsilon-coresets for clustering (with outliers) in doubling metrics.
\newblock In \emph{59th Annual Symposium on Foundations of Computer Science
  (FOCS)}, pp.\  814--825. IEEE, 2018.

\bibitem[Huang et~al.(2019)Huang, Jiang, and Vishnoi]{DBLP:conf/nips/HuangJV19}
Huang, L., Jiang, S.~H., and Vishnoi, N.~K.
\newblock Coresets for clustering with fairness constraints.
\newblock In \emph{NeurIPS}, pp.\  7587--7598, 2019.

\bibitem[Jain \& Vazirani(2001)Jain and Vazirani]{DBLP:journals/jacm/JainV01}
Jain, K. and Vazirani, V.~V.
\newblock Approximation algorithms for metric facility location and
  \emph{k}-median problems using the primal-dual schema and lagrangian
  relaxation.
\newblock \emph{J. {ACM}}, 48\penalty0 (2):\penalty0 274--296, 2001.

\bibitem[Jain et~al.(2002)Jain, Mahdian, and Saberi]{DBLP:conf/stoc/JainMS02}
Jain, K., Mahdian, M., and Saberi, A.
\newblock A new greedy approach for facility location problems.
\newblock In \emph{{STOC}}, pp.\  731--740. {ACM}, 2002.

\bibitem[Kearns \& Vazirani(1994)Kearns and Vazirani]{kearns1994introduction}
Kearns, M.~J. and Vazirani, U.~V.
\newblock \emph{An introduction to computational learning theory}.
\newblock MIT press, 1994.

\bibitem[Kloks(1994)]{kloks1994treewidth}
Kloks, T.
\newblock \emph{Treewidth: computations and approximations}, volume 842.
\newblock Springer Science \& Business Media, 1994.

\bibitem[Langberg \& Schulman(2010)Langberg and
  Schulman]{DBLP:conf/soda/LangbergS10}
Langberg, M. and Schulman, L.~J.
\newblock Universal epsilon-approximators for integrals.
\newblock In \emph{{SODA}}, pp.\  598--607. {SIAM}, 2010.

\bibitem[Maalouf et~al.(2019)Maalouf, Jubran, and
  Feldman]{DBLP:conf/nips/MaaloufJF19}
Maalouf, A., Jubran, I., and Feldman, D.
\newblock Fast and accurate least-mean-squares solvers.
\newblock In \emph{NeurIPS}, pp.\  8305--8316, 2019.

\bibitem[Maniu et~al.(2019)Maniu, Senellart, and Jog]{DBLP:conf/icdt/ManiuSJ19}
Maniu, S., Senellart, P., and Jog, S.
\newblock An experimental study of the treewidth of real-world graph data.
\newblock In \emph{{ICDT}}, volume 127 of \emph{LIPIcs}, pp.\  12:1--12:18.
  Schloss Dagstuhl - Leibniz-Zentrum fuer Informatik, 2019.

\bibitem[Marom \& Feldman(2019)Marom and Feldman]{DBLP:conf/nips/MaromF19}
Marom, Y. and Feldman, D.
\newblock $k$-means clustering of lines for big data.
\newblock In \emph{NeurIPS}, pp.\  12797--12806, 2019.

\bibitem[Munteanu \& Schwiegelshohn(2018)Munteanu and
  Schwiegelshohn]{DBLP:journals/ki/MunteanuS18}
Munteanu, A. and Schwiegelshohn, C.
\newblock Coresets-methods and history: {A} theoreticians design pattern for
  approximation and streaming algorithms.
\newblock \emph{{KI}}, 32\penalty0 (1):\penalty0 37--53, 2018.

\bibitem[{OpenStreetMap contributors}(2020)]{OpenStreetMap}
{OpenStreetMap contributors}.
\newblock {Planet dump retrieved from https://planet.osm.org }.
\newblock \url{ https://www.openstreetmap.org }, 2020.

\bibitem[Phillips(2017)]{Phillips16}
Phillips, J.~M.
\newblock Coresets and sketches.
\newblock In Toth, C.~D., O'Rourke, J., and Goodman, J.~E. (eds.),
  \emph{Handbook of discrete and computational geometry}, chapter~48. Chapman
  and Hall/CRC, 3rd edition, 2017.
\newblock \doi{10.1201/9781315119601-48}.

\bibitem[Rattigan et~al.(2007)Rattigan, Maier, and Jensen]{rattigan2007graph}
Rattigan, M.~J., Maier, M., and Jensen, D.
\newblock Graph clustering with network structure indices.
\newblock In \emph{Proceedings of the 24th international conference on Machine
  learning}, pp.\  783--790. ACM, 2007.

\bibitem[Robertson \& Seymour(1986)Robertson and Seymour]{robertson1986graph}
Robertson, N. and Seymour, P.~D.
\newblock Graph minors. {II}. algorithmic aspects of tree-width.
\newblock \emph{Journal of algorithms}, 7\penalty0 (3):\penalty0 309--322,
  1986.

\bibitem[Sack \& Urrutia(1999)Sack and Urrutia]{sack1999handbook}
Sack, J.-R. and Urrutia, J.
\newblock \emph{Handbook of Computational Geometry}.
\newblock Elsevier, 1999.

\bibitem[Schmidt et~al.(2018)Schmidt, Schwiegelshohn, and
  Sohler]{DBLP:journals/corr/abs-1812-10854}
Schmidt, M., Schwiegelshohn, C., and Sohler, C.
\newblock Fair coresets and streaming algorithms for fair k-means clustering.
\newblock \emph{CoRR}, abs/1812.10854, 2018.

\bibitem[Shekhar \& Liu(1997)Shekhar and Liu]{DBLP:journals/tkde/ShekharL97}
Shekhar, S. and Liu, D.
\newblock {CCAM:} {A} connectivity-clustered access method for networks and
  network computations.
\newblock \emph{{IEEE} Trans. Knowl. Data Eng.}, 9\penalty0 (1):\penalty0
  102--119, 1997.

\bibitem[Sohler \& Woodruff(2018)Sohler and Woodruff]{DBLP:conf/focs/SohlerW18}
Sohler, C. and Woodruff, D.~P.
\newblock Strong coresets for $k$-median and subspace approximation: Goodbye
  dimension.
\newblock In \emph{{FOCS}}, pp.\  802--813. {IEEE} Computer Society, 2018.

\bibitem[Tansel et~al.(1983)Tansel, Francis, and Lowe]{tansel1983state}
Tansel, B.~C., Francis, R.~L., and Lowe, T.~J.
\newblock State of the art—location on networks: a survey, part i and ii.
\newblock \emph{Management Science}, 29\penalty0 (4):\penalty0 482--497, 1983.

\bibitem[Thorup(2005)]{Thorup05}
Thorup, M.
\newblock Quick $k$-median, $k$-center, and facility location for sparse
  graphs.
\newblock \emph{SIAM J. Comput.}, 34\penalty0 (2):\penalty0 405--432, 2005.
\newblock ISSN 0097-5397.
\newblock \doi{10.1137/S0097539701388884}.

\bibitem[Vapnik \& Chervonenkis(2015)Vapnik and
  Chervonenkis]{vapnik2015uniform}
Vapnik, V.~N. and Chervonenkis, A.~Y.
\newblock On the uniform convergence of relative frequencies of events to their
  probabilities.
\newblock In \emph{Measures of complexity}, pp.\  11--30. Springer, 2015.

\bibitem[Varadarajan \& Xiao(2012)Varadarajan and
  Xiao]{varadarajan2012sensitivity}
Varadarajan, K. and Xiao, X.
\newblock On the sensitivity of shape fitting problems.
\newblock In \emph{IARCS Annual Conference on Foundations of Software
  Technology and Theoretical Computer Science (FSTTCS 2012)}. Schloss
  Dagstuhl-Leibniz-Zentrum fuer Informatik, 2012.

\bibitem[Yiu \& Mamoulis(2004)Yiu and Mamoulis]{DBLP:conf/sigmod/YiuM04}
Yiu, M.~L. and Mamoulis, N.
\newblock Clustering objects on a spatial network.
\newblock In \emph{{SIGMOD} Conference}, pp.\  443--454. {ACM}, 2004.

\end{thebibliography}
	\bibliographystyle{icml2020}
	
	\ifeightpage
	\else 

	\newpage
	\appendix
	
	\onecolumn
\section{Size Lower Bounds}
We present an $\Omega(\frac{k}{\epsilon} \cdot \tw(G))$ lower bound for clustering in graphs, which matches the linear dependence on $\tw(G)$ in our coreset construction.
Previously, only very few lower bounds were known for coresets.
For \kCenter in $d$-dimensional Euclidean spaces, it was known that size $\Omega(\frac{k}{\epsilon^d})$ is required [D. Feldman, private communication]. Recent work~\cite{DBLP:conf/icml/BravermanJKW19}
proved an $\Omega(\log n)$ lower bound for simultaneous coresets in Euclidean spaces,
where a simultaneous coreset is a single coreset that is simultaneously an $\epsilon$-coreset for multiple objectives such as \kMedian and \kCenter.
However, no lower bounds for \kMedian were known.
In fact, even the $O(\log n)$ factor for general metrics was not justified.
Since our hard instance in Theorem~\ref{thm:tw_kz_lb} consists of $O(\frac{k}{\epsilon} \cdot 2^t)$ vertices, it readily implies for the first time that the $O(\log n)$ factor is optimal for general metrics (see Corollary~\ref{cor:lb_general}).

Our lower bound is actually split into two theorems: one for the tree case ($\tw(G) = 1$) and one for the other cases ($\tw(G) \geq 2$).
Ideally, we would use a unified argument, but unfortunately the general argument for $\tw(G)\geq 2$ does not apply in the special case $\tw(G) = 1$ because some quantity is not well defined, and we thus need to employ a somewhat different argument for the tree case.

\begin{theorem}[Lower Bound for Graphs with Treewidth $\geq 2$]
	\label{thm:tw_kz_lb}
	For every $0 < \epsilon < 1$ and integers $t, k \geq 1$,
	there exists an unweighted graph $G=(V, E)$ with $\tw(G) \leq t + 1$,
	such that any $\epsilon$-coreset for \kMedian on data set $X = V$ has size $\Omega(\frac{k}{\epsilon} \cdot t)$.
\end{theorem}
\begin{proof}
	The vertex set of $G_{k,\eps}$ is defined as $L \cup R \cup \{ u_0 \}$.
	Let $m := \frac{k}{\epsilon}$. Both $L$ and $R$ consist of $m$ groups,
	i.e. $L := \bigcup_{i=1}^{m}{L_i}$ and $R := \bigcup_{i = 1}^{m}{R_i}$.
	For $i \in [m]$, $L_i$ consists of $t$ elements, and $R_i$ consists of $2^t$ elements.
	Let $L_i := \{ l^{(i)}_j : j \in [t]\}$,
	and
	$R_i := \{ r^{(i)}_J : J \subseteq [t] \}$.
Since $t\geq 1$, $L$ is non-empty.
Define a special connection point $u_0$ to which all points of $L\cup R$ connect to (the specific way of connection is defined in the next paragraph).
	
	The edge set is defined as follows. All edges are of weights $1$.
	Connect all points in $L \cup R$ to $u_0$.
	For each $i\in [m]$, for each $l^{(i)}_j \in L_i$ and $r^{(i)}_J \in R_i$,
	if $j \in J$, add an edge $\{ l^{(i)}_j, r^{(i)}_J \}$.
Finally, let $T = \Omega(\frac{m}{k}\cdot 2^t)$, and make $T-1$
	copies of each point in $L$,
	which we call shadow vertices:
	for each $l^{(i)}_j$, create $T-1$ vertices, and connect them to $l^{(i)}_j$ directly (so they form a star with center $l^{(i)}_j$).

	\begin{fact}
		All distances in $G_{k,\eps}$ are $2$, except that the distances between $L_i$ and $R_i$ ($i \in [m]$) are $1$.
	\end{fact}
	
	For simplicity, we use $G$ to represent $G_{k,\eps}$ in the following.
	
	\paragraph{Treewidth Analysis}
	First, consider removing $u_0$ from $G$, and define the resultant graph as $G'$.
	Then $\tw(G) \leq \tw(G') + 1$.
	Observe that $G'$ has $m$ components: $\{L_i \cup R_i : i \in [m]\}$,
	so it suffices to bound the treewidth for each component.
	For each such component, since removing $L_i$ makes all points in the component isolated, we conclude that the treewidth of the component is at most $|L_i| = t$.
	Therefore, we conclude that $\tw(G) \leq t + 1$.
	\paragraph{Error Analysis}
	Suppose $D \subseteq V$ (with weight $w$) is an $O(\epsilon)$-coreset of size $o(\frac{k}{\epsilon} \cdot t)$.
	By manipulating the weight $w$, we assume w.l.o.g. that $D$ does not contain the shadow vertices.
	Pick any $k$-subset $S \subseteq [m]$, such that for every $i \in S$, $|D \cap L_i| \leq \frac{t}{2}$ and $|D \cap R_i| \leq \frac{t}{2}$.
	Such $S$ must exist, since otherwise there would be $m - k = \Omega(\frac{k}{\epsilon})$ number of $i$'s, such that $|D\cap L_i| + |D \cap R_i| > \frac{t}{2}$,
	which contradicts $|D| = o(\frac{k}{\epsilon} \cdot t)$.
	
	We would then pick two subsets $P_i, Q_i \subseteq [t]$ for each $i \in S$, which correspond to two points in $R_i$ and encode two subsets of $L_i$, as in the following claim.
	
	\begin{claim}
		For each $i \in S$, there exists $P_i, Q_i \subseteq [t]$,
		such that
		\begin{enumerate}
			\item If $l^{(i)}_j \in D \cap L_i$, then $j \in P_i$ and $j \in Q_i$.
			\item If $r^{(i)}_J \in D \cap R_i$,
			then $J \neq P_i$ and $J \neq Q_i$.
			\item $|P_i| \leq |D \cap L_i| + O(1)$, and $|Q_i| \geq t - O(1)$.
		\end{enumerate}
	\end{claim}
	\begin{proof}
Suppose $D \cap R_i = \{ r^{(i)}_{J_1}, \ldots, r^{(i)}_{J_s}\}$, and let $\mathcal{J} = \{J_1, \ldots, J_s\}$.
Find the minimum cardinality $P_i$ such that item 1 and 2 holds:
		this is equivalent to find the smallest $P'$, such that $(D \cap L_i) \cup P' \notin \mathcal{J}$.
Such $P'$ may be found in a greedy way: try out all $0$-subsets, 1-subsets, \ldots, until $(D \cap L_i) \cap P' \notin \mathcal{J}$.
Since $|D \cap R_i| \leq \frac{t}{2}$ and $|D\cap L_i| \leq \frac{t}{2}$, such greedy procedure must end after trying out $O(1)$-subsets and hence $|P_i|\leq |D\cap L_i|+O(1)$.

		Let $Q_i$ denote the set with the maximum cardinality such that item 1 and 2 holds.
By a similar argument, we can prove that $|Q_i|\geq t-O(1)$.
	\end{proof}
	Based on this claim, we define $C_1 := \{ r^{(i)}_{P_i} : i \in S \}$, and $C_2 := \{ r^{(i)}_{Q_i} : i \in S\}$.
	Observe that the cost on both $C_1$ and $C_2$ are the same on the coreset $D$ (by item 1).
	However, the objective on $C_1$ and $C_2$ differ by an $\Omega(\epsilon)$ factor (where we use item 2 and 3).
	To see it,
	\begin{align*}
	\cost(V, C_1) = &\underbrace{2}_{u_0} + \underbrace{2(m - k)\cdot T\cdot t}_{\text{cost of $L\setminus L_i$}} + \underbrace{2(m\cdot 2^t - k)}_{\text{cost of $R$}} + \underbrace{T\cdot \sum_{i \in S}{ 2(t - |P_i|) + |P_i|}}_{\text{cost of $L_i$}} \\
	= &2 + 2(m-k)\cdot T \cdot t + 2(m\cdot 2^t - k) + 2 ktT - (2-1)T \cdot\sum_{i\in S}{|P_i|} \\
	\geq &2 + 2(m-k)\cdot T \cdot t + 2(m\cdot 2^t - k) + 2 ktT - kT \cdot (\frac{t}{2} + O(1)).
	\end{align*}
	Similarly,
	\begin{align*}
	\cost(V, C_2) = &2 + 2 (m - k)\cdot T\cdot t + 2 (m\cdot 2^t - k) + T\cdot \sum_{i \in S}{2 (t - |Q_i|) + |Q_i|} \\
	\leq &2 + 2 (m-k)\cdot T \cdot t + 2 (m\cdot 2^t - k) + 2 ktT - kT\cdot (t - O(1))
	\end{align*}
	So,
	\begin{align*}
	\frac{\cost(V, C_1)}{\cost(V, C_2)}
	&\geq 1 + \frac{ kT\cdot (\frac{t}{2} - O(1))}
	{2 + 2(m-k)\cdot T \cdot t + 2(m\cdot 2^t - k) + 2ktT - kT\cdot (t - O(1))} \\
	&\geq 1 + \Omega(\epsilon)
	\end{align*}

    where the last inequality is by $m=\frac{k}{\epsilon}$ and $T=\Omega(\frac{m}{k}\cdot2^t)$.
	This contradicts the fact that $D$ is an $O(\epsilon)$-coreset.
\end{proof}

Then we prove for the special case with treewidth $1$, which is the tree case.

\begin{theorem}[Lower Bound for Star Graphs]
	\label{thm:star_lb}
	For every $0 < \epsilon < 1/3$ and integer $k \geq 1$,
	there exists an (unweighted) start graph $G=(V, E)$ with $|V|=O(\frac{k}{\epsilon})$
	such that
	any $\epsilon$-coreset for \kMedian on data set $X = V$ has size $\Omega(\frac{k}{\epsilon})$.
\end{theorem}

\begin{proof}
	Denote the root node of the star graph $G=(V,E)$ by $r$ and leaf nodes by $x_1,\ldots,x_n$ ($n\geq \frac{100 k}{\epsilon}$).
Suppose $D \subseteq V$ (with weight $w$) is an $\eps$-coreset of size $o(\frac{k}{\eps})$.
Let $W=\sum_{x\in D\setminus \{r\}}w(x)$.
Consider a $k$-center set where all centers are on $r$.
We have that
	\begin{equation}
	\label{eq:W_bound}
	W= \sum_{x\in D} w(x)\cdot d(x,r)\geq (1-\eps)\cdot \cost(X,r) = (1-\eps)\cdot n,
	\end{equation}
where the inequality is from the fact that $D$ is an $\eps$-coreset.

	Next, we construct two center sets $C_1$ and $C_2$.
Let $C_1\subseteq V\setminus (D\cup \left\{r\right\})$ be a collection of $k$ distinct leaf nodes that are not in $D$.
Let $C_2$ be the collection of $k$ nodes in $D\setminus \left\{r\right\}$ with largest weights.
By construction, we have that
	\begin{equation}
	\label{ineq:weight_bound}
	W'=\sum_{x\in C_2} w(x)\geq \frac{k W}{|D|}\geq 100 \eps W\geq 50 \eps n,
	\end{equation}
	where the last inequality is by Inequality~\eqref{eq:W_bound}.
Moreover, since $D$ is an $\eps$-coreset, we have that
	\begin{equation}
	\label{ineq:cluster_bound}
	2 W+w(r)=\cost(D,C_1)\leq (1+ \eps)\cdot \cost(V,C_1) \leq 2\cdot (2 n+1).
	\end{equation}
By symmetry, $\cost(V,C_1)=\cost(V,C_2)$.
Then by the definition of coreset, we have
	\[
	\frac{\cost(D,C_1)}{\cost(D,C_2)}=\frac{2 W+w(r)}{2(W-W')+w(r)}\leq \frac{1+\eps}{1-\eps}.
	\]
However, by Inequalities~\eqref{ineq:weight_bound} and~\eqref{ineq:cluster_bound}, we have
	\begin{eqnarray*}
	\begin{split}
	\frac{2 W+w(r)}{2(W-W')+w(r)} &\geq && \frac{2 W+w(r)}{2W+w(r) - 50\cdot 2 \eps n} && (\text{Ineq.~\eqref{ineq:weight_bound}}) \\
	&\geq &&\frac{2\cdot (2 n+1)}{2\cdot (2 n+1) - 100 \eps n} && (\text{Ineq.~\eqref{ineq:cluster_bound}}) \\
	& >&& \frac{1+\eps}{1-\eps}, &&
	\end{split}
	\end{eqnarray*}
	which is a contradiction.
This completes the proof.
\end{proof}
Combining Theorems~\ref{thm:star_lb} and~\ref{thm:tw_kz_lb}, we obtain a lower bound of $\Omega(\frac{k}{\eps}\cdot \tw(G))$ for the coreset size.
Moreover, we observe that the hard instance in Theorem~\ref{thm:tw_kz_lb} has $O(2^{t})$ nodes,
which in fact implies an $\Omega(\log n)$ size lower bound for general graphs. We state this corollary as follows.
\begin{corollary}
	\label{cor:lb_general}
	For every $0 < \epsilon < 1$ and integers $n, k \geq 1$, there exists an unweighted graph $G = (V, E)$ with $|V| = O(n)$ such that any $\epsilon$-coreset for \kMedian on data set $X = V$ has size $\Omega(\frac{k}{\epsilon} \cdot \log n)$.
\end{corollary}

\subsection{Lower Bound for Coresets in 1D Lines}
\label{sec:lb_1d}
\begin{theorem}[Lower Bound for 1D Lines]
	\label{1d1c}
	For every $0<\epsilon<1/24$ and integer $k\geq 1$, there exists a set of data points $V\subseteq \mathbb{R}$, such that every $\epsilon$-coreset for \kMedian of $V$ has size $\Omega(\frac{k}{\sqrt{\epsilon}})$, even if the coreset may use any point in $\mathbb{R}$.
\end{theorem}

We first introduce our technical Lemma \ref{lem:1d_tech}, which shows a quadratic function cannot be approximated by an affine linear function in a short interval. We note that a similar argument appears in \cite{DBLP:conf/icml/BravermanJKW19}, which shows the function $f(x)=\sqrt{x}$ cannot be approximated by an affine linear function in a short interval.

\begin{lemma} \label{lem:1d_tech}
	Let $\eps\in (0,1/12)$ and $1/2\leq p\leq q\leq 1$. Suppose $f(x)=ax^2+b$ where $a>0,b$, $g(x)$ is a non-negative linear function on $[p,q]$, and $g(x)\in (1\pm \eps)f(x)$ for every $x\in [p,q]$, then $q\leq p+\sqrt{24(b/a+1)\eps}$.
\end{lemma}

\begin{proof}
	
	Since $g(x)$ is non-negative and $g(x)\in (1\pm \eps)f(x),\forall x\in [p,q]$, we have that $\int_{p}^q g(x)dx\in (1\pm\eps)\int_{p}^q f(x) dx$. By computation, $\int_{p}^q f(x)dx=\int_{p}^q (ax^2+b) dx=\frac{a(q^3-p^3)}{3}+b(q-p)$. Since $g(x)$ is linear, $\int_p^q g(x)dx=(g(p)+g(q))(q-p)/2$. By the fact that $g(x)\in (1\pm\eps)f(x)=(1\pm\eps)(ax^2+b )$, we have $\int_a^b g(x)dx\in (1\pm\eps) (ap^2+aq^2+2b)(q-p)/2$. But $\int_a^b g(x)dx\leq (1+\eps)\int_a^b (ax^2+b) dx=(1+\eps)(a(q^3-p^3)/3+b(q-p)) $, so we have that,
	
	$$\frac{a(q^3-p^3)/3+b(q-p)}{(ap^2+aq^2+2b)(q-p)/2}\geq \frac{1-\eps}{1+\eps}\geq 1-2\eps.$$
	
	So we have $p^2+q^2-\frac{2pq}{1-6\eps}-\frac{12b\eps}{(1-6\eps)a}\leq 0$. Since $\frac{1}{2}\leq p\leq q\leq 1$ and $0<\eps<1/12$, we have that
	$$
	(q-p)^2\leq \frac{12b\eps}{(1-6\eps)a}+\frac{12\eps pq}{1-6\eps}\leq 24(b/a+1)\eps
	$$
	
	which implies $q\leq p+\sqrt{24(b/a+1)\eps}$.
\end{proof}

Now we are ready to prove Theorem~\ref{1d1c}.

\begin{proof}[Proof of Theorem~\ref{1d1c}]
	
	We first prove the basic case $k=1$. Without loss of generality, we can assume $V=[-1,1]$.\footnote{For discretization, we can let $V=\{0,\pm\frac{1}{m},...,\pm 1\}$ for large enough $m$.} Let $f(x):=\mathrm{cost}(V,\{x\})=\int_V \|t-x\|dt$ denote the cost of connecting $V$ to $x\in \mathbb{R}$ (the \ProblemName{1-Median} value). Note that $f(x)=x^2+1$ on $[-1,1]$.
	
	Assume $D$ is an $\eps$-coreset of $V$ for the $1$-median problem. Recall that $D$ is a weighted set and may contain the ambient points of the real line. Let $g(x):=\mathrm{cost}(D,\{x\})$. Then $g(x)$ is a piecewise linear function and the transition from one affine linear function to another happens only when $x$ crosses a coreset point. So the number of pieces is at most $|D|+1$. We need to prove $g(x)$ has at least $\Omega(\frac{1}{\sqrt{\eps}})$ pieces.
	
	Since $D$ is an $\eps$-coreset of $V$, we have that $g(x)\in (1\pm\eps)f(x)$ for every $x\in [-1,1]$. Assume $g(x)$ has $m$ pieces in $[1/2,1]$ and their connecting points are $x_0=1/2<x_1<...<x_m=1$. Then $g(x)$ is affine linear in $[x_{i-1},x_i]$ but $g(x)\in (1\pm\eps)f(x)$, so by Lemma \ref{lem:1d_tech}, $x_i<x_{i-1}+\sqrt{24(1/1+1)\eps}<x_{i-1}+7\sqrt\eps$. So $m\geq \frac{1/2}{7\sqrt\eps}=\frac{1}{14\sqrt\eps}$.
	
	Now we consider the case of general $k$. We put $k$ copies of $[-1,1]$ in the real line. In particular, we let $V_i=[-1+(i-1)t,1+(i-1)t]$ for a large enough positive number $t>\frac{39}{\sqrt\eps}$ and $V=\cup_{i=1}^k V_i$. Let $S_i=[1/2+(i-1)t,1+(i-1)t]$ be a subset of $V_i$. We partition each $S_i$ into $m=\Theta(1/\sqrt{\eps})$ intervals $S_{i1},...,S_{im}$, such that each of them has length $13\sqrt{\eps}$.
	
	Now, for the sake of contradiction, we assume there is an $\eps$-coreset $D$ of $V$ for the $k$-median problem, such that $|D|<mk/2$.
By averaging, we know that there is a $j^*\in [k]$ such that $\cup_{i=1}^k S_{ij^*}$ contains at most $k/2$ points of $D$. Let $b=|D\cap (\cup_{i=1}^k S_{ij^*})|$ in the following, then $b<k/2$.
So there are $k-b>k/2$ many $i\in [k]$ such that $S_{ij^*}$ doesn't contain any coreset point. We assume $S_{0j^*}=[a_0,b_0]$ and let $x$ be a variable in $[a_0,b_0]$. We construct the following set of centers $C_x=\{x_i:i\in [k]\}$ where if $S_{ij^*}\cap D\not=\emptyset$, $x_i=x+(i-1)t$, otherwise $x_i=(i-1)t$.
	
	Let $f(x)=\mathrm{cost}(V,C_x)$, then $f(x)=b+(k-b)(x^2+1)=(k-b)x^2+k$. We note that $k\leq f(x)\leq 2k$ on $[a_0,b_0]$.
Let $D_i=D\cap [(i-4/3)t,(i-2/3)t]$, $D'=\cup_{i\in[k]} D_i$ and $D''=D\setminus D'$. We first claim that $|\mathrm{cost}(D'',C_x)- \mathrm{cost}(D'',C_{a_0})|\leq  O(\sqrt{\eps}k/t)$. Actually, note that when $x\in [a_0,b_0]$, the connection cost of points in $D''$ is always $\Omega(t)$, along with the fact that when $x\in [a_0,b_0]$, $\mathrm{cost}(D'',C_x)\leq \mathrm{cost}(D,C_x)\leq (1+\eps)\mathrm{cost}(V,C_x)\leq 3k$, we know that the total weight of $D''$ is at most $3k/t$. Now, since $x$ changes by at most $13\sqrt{\eps}$, the connection cost of every point in $D''$ changes by at most $13\sqrt{\eps}$, so we have
	$$|\mathrm{cost}(D'',C_x)- \mathrm{cost}(D'',C_{a_0})|\leq \frac{39\sqrt{\eps}k}{t}.$$
	
	Let $g(x)=\mathrm{cost}(D',C_x)+\mathrm{cost}(D'',C_{a_0})$. Since $t>\frac{39}{\sqrt\eps}$ and $f(x)\geq k$ for every $x\in[a_0,b_0]$, we conclude that when $x\in [a_0,b_0]$, $$
	|g(x)-\mathrm{cost}(D,C_x)|=|\mathrm{cost}(D'',C_{a_0})-\mathrm{cost}(D'',C_x)|\leq \frac{39\sqrt{\eps}k}{t}\leq \eps f(x).
	$$
	But $\mathrm{cost}(D,C_x)\in (1\pm\eps)f(x)$ on $[a_0,b_0]$, so we have that $g(x)\in (1\pm 2\eps)f(x)$ on $[a_0,b_0]$.
	
	Now we show that $g(x)$ is an affine linear function in $[a_0,b_0]$. We note that $D'=\cup_{i\in [k]} D_i$ and $\mathrm{cost}(D_i,C_x)=\mathrm{cost}(D_i,\{x_i\})$ for any $x\in [a_0,b_0]$. If $x_i=x+(i-1)t$ then by construction, $x_i$ never crosses any coreset point in $D_i$, so $\mathrm{cost}(D_i,\{x_i\})$ remains affine linear. On the other hand, if $x_i=(i-1)t$, then $\mathrm{cost}(D_i,\{x_i\})$ is a constant. So we know that $g(x)=\mathrm{cost}(D',C_x)+\mathrm{cost}(D'',C_{a_0})=\sum_{i\in [k]} \mathrm{cost}(D_i,\{x_i\})+\mathrm{cost}(D'',C_{a_0})$ is an affine linear function on $[a_0,b_0]$. But $g(x)\in (1\pm 2\eps)f(x)$ on $[a_0,b_0]$, by Lemma \ref{lem:1d_tech}, we know that the length of $[a_0,b_0]$ is at most $\sqrt{24(\frac{k}{k-b}+1)\cdot 2\eps}\leq \sqrt{144\eps}=12\sqrt{\eps}$,
arriving at a contradiction.
\end{proof}

 \fi 

\end{document}